\documentclass[journal, a4paper]{IEEEtran}
\IEEEoverridecommandlockouts
\usepackage{amsfonts}
\usepackage{dsfont}
\usepackage[normalem]{ulem}
\usepackage{setspace}
\usepackage{color}
\usepackage{amssymb}
\usepackage{cite}
\usepackage[cmex10]{amsmath}
\usepackage{algorithm}
\usepackage{algorithmic} 
\usepackage{array}
\usepackage{mathrsfs}
\usepackage{graphicx}
\usepackage{latexsym}
\usepackage{amscd}
\usepackage{amsfonts}
\usepackage{subfigure}
\usepackage{amsmath,amscd,amssymb,verbatim}
\usepackage{graphics}
\usepackage{amsthm}
\usepackage{tabularx}
\usepackage[utf8]{inputenc}
\usepackage{authblk}
\usepackage{amsmath,amsthm}
\usepackage{url} 
\usepackage[utf8]{inputenc}
\usepackage[left=1.6cm,right=1.6cm,top=1.9cm,bottom=3.5cm]{geometry}

\newtheorem{prop}{Proposition}
\IEEEoverridecommandlockouts
\usepackage[textsize=tiny,colorinlistoftodos]{todonotes}
\makeatletter
\define@key{todonotes}{YLH}[]{
	\setkeys{todonotes}{author=\textbf{Yulin}, color=red!30}}
\makeatother
\usepackage{color}

\allowdisplaybreaks

\begin{document}

\title{Efficient Trajectory Design and Communication Scheduling for Dual-UAV Jamming-Aided Secure Communication Networks}
 \author[$\dag$]{Xinran Wang, Peng Wu, Xiaopeng Yuan, Yulin Hu$^*$, and Anke Schmeink
 \thanks{X. Wang, P. Wu, and Y. Hu are with School of Electronic Information, Wuhan University, 430072 Wuhan, China. 
(Email: $xinran.wang|peng.wu|yulin.hu$@whu.edu.cn) $^*$Y. Hu is the corresponding author.
 }
 \thanks{X. Yuan and A. Schmeink are with Chair of Information Theory and Data Analytics, RWTH Aachen University, 52074 Aachen, Germany. (Email: $yuan|schmeink$@inda.rwth-aachen.de)}
  }

\maketitle
\begin{abstract}
\hyphenpenalty=8000
\tolerance=1800 
We study dual-unmanned aerial vehicle (UAV) jamming-aided secure communication networks, in which one UAV delivers confidential data to multiple ground users (GUs), while a cooperative UAV provides protective interference against a ground eavesdropper. To enforce fairness, we maximize the minimum secrecy throughput across GUs by jointly designing trajectories and communication scheduling. The key difficulty lies in the continuous-time nature of UAV trajectories and the tight space-time coupling between the transmitter and the jammer, which jointly render the problem infinite-dimensional and nonconvex. To address these challenges, we characterize, for the first time, the structure of the optimal trajectories and rigorously prove that they follow a collaborative successive hover-and-fly (co-SHF) structure, where the two UAVs visit a limited number of synchronized co-hovering point pairs, and during each flight segment at least one UAV moves at maximum speed. Leveraging this structure, we reformulate the problem into a finite-dimensional form, without loss of optimality, over hovering and turning points, hovering durations, and scheduling. For tractability, we adopt a minimum-distance approximation of continuous anti-collision constraints and employ concave lower bounds on secrecy throughput within a successive convex approximation (SCA) method, which converges and, thanks to the co-SHF reduction in optimization variables and constraints, achieves low computational complexity. Numerical results show that, compared with time-discretization and no-jamming benchmarks, the proposed co-SHF design improves the min-secrecy and user fairness while requiring significantly less runtime.

\end{abstract}

\section{Introduction}
 
As the demand for more adaptive, intelligent, and secure communication infrastructures grows in future wireless networks such as 6G, unmanned aerial vehicle (UAV)-assisted communication is expected to serve as a flexible and intelligent component in next-generation system design \cite{Banafaa2024,Chen2023,Geraci2022}. 
With their high mobility, flexible deployment, and line-of-sight (LoS)-dominated air-to-ground (A2G) channels, UAVs have been increasingly integrated into wireless networks as aerial base stations, aerial relays, and aerial access points \cite{Xiao2022,Xu2020_1}. 
In particular, by appropriately designing UAV trajectories, possibly jointly with resource allocation, the spatial-temporal degrees of freedom offered by UAV mobility can be exploited to significantly improve network performance metrics such as throughput, coverage, and latency \cite{Pervez2024,Yan2024,Chang2023}. These features have driven extensive interest in UAV trajectory design for communication networks.

While the LoS-dominated nature of A2G channels facilitates wireless transmission by reducing path loss and signal blockage, it simultaneously gives rise to inherent security vulnerabilities due to the broadcast characteristics of wireless propagation. As a result, UAV-assisted communications are particularly susceptible to unauthorized interception and surveillance, posing serious threats to information confidentiality \cite{Wang2019,Wang2022,Sun2019}. Ensuring secure communication is therefore a fundamental requirement for UAV-enabled wireless networks. Considering the stringent energy constraints and high mobility of UAV platforms, physical-layer security (PLS) has emerged as an attractive solution for enhancing communication security without relying on additional cryptographic overhead \cite{Xiao2024,Mozaffari2021,Hassija2021}. In UAV-enabled networks, the high controllability of UAV motion further enables trajectory design and resource allocation to be exploited as effective means for improving secrecy performance \cite{Zhou2020,Na2022,Wang2024,Huang2025}. 
To further strengthen security, artificial noise (AN)-based cooperative jamming has been widely investigated as a mechanism to deliberately degrade the eavesdropper’s reception while preserving the quality of legitimate links. By embedding jamming signals into the system design and coordinating transmission and interference in space and time, jamming-aided approaches provide additional degrees of freedom for secure communication enhancement in UAV-assisted networks.

To further enhance communication security in UAV-assisted networks, cooperative architectures with two coordinated UAVs that play heterogeneous roles, where one UAV acts as a transmitter and the other serves as a mobile jammer, have recently attracted increasing attention \cite{Zhong2019}. By separating transmission and jamming functionalities across two aerial platforms, dual-UAV cooperative jamming offers enhanced spatial flexibility, which can substantially improve secrecy performance.
Existing studies on dual-UAV secure communications have considered various objectives and system settings, including robust designs under multiple eavesdroppers, propulsion-aware operation, and joint trajectory and resource optimization under practical constraints \cite{Miao2020,Cai2020,Lei2023}. 

Despite the notable secrecy gains achieved by existing works through joint trajectory and resource optimization \cite{Zhou2020,Na2022,Miao2020,Cai2020,Lei2023}, most of them handle continuous-time trajectory design via \textit{time discretization} (TD), where the mission duration is divided into a number of short time slots and the UAV position is approximated as constant within each slot \cite{Guo2021,Zhang2019}. To approach the optimal continuous-time solution, TD-based methods require increasingly fine discretization, which substantially enlarges the number of optimization variables and constraints and leads to high computational complexity. On the other hand, coarse discretization can significantly reduce complexity but inevitably introduces discretization errors, resulting in suboptimal solutions. This inherent accuracy-complexity tradeoff becomes particularly severe in multi-UAV systems. As the number of UAVs increases, the optimization variables and coupling constraints scale in a multiplicative manner across time slots and vehicles, causing the computational burden of TD-based formulations to grow rapidly with both the trajectory resolution and the network size. Such poor scalability motivates structure-exploiting continuous-time formulations that avoid dense discretization while preserving optimality.

In dual-UAV cooperative jamming systems, identifying an exploitable optimal coordination structure is particularly challenging yet crucial for scalable design. Unlike single-UAV settings, secrecy performance depends jointly on the transmitter and jammer positions, and the two role-asymmetric UAVs must coordinate their motion in both space and time while continuously satisfying collision-avoidance constraints and jointly designing resource allocation. As a result, the optimal coordinated trajectories cannot be characterized by a straightforward scaling-up of existing TD-based formulations, nor by a direct reuse of single-UAV structural results. 
Existing structure-based and analytical continuous-time trajectory design approaches for single-UAV systems, such as the successive hover-and-fly (SHF) structure \cite{Yuan2021,Yuan2023,Wang2024} and artificial potential field (APF)-based methods \cite{Huang2025}, provide useful insights into avoiding dense time discretization. However, they do not address the above dual-UAV coordination requirements, such as when and where two role-asymmetric UAVs should synchronize their hovering, or how their flying speeds should be jointly coordinated between hovering phases. To the best of our knowledge, a rigorous characterization of the optimal coordinated trajectories for dual-UAV transmitter-jammer systems, together with an efficient structure-based design framework, remains unavailable. 

\begin{table*}[t]
\centering
\caption{Comparison with Existing Works.}
\label{tab:comparison}
\renewcommand{\arraystretch}{1.2}
\begin{tabular}{c|c|c|c|c}
\hline
\textbf{Work} 
& \textbf{UAV Settings} 
& \textbf{System Roles}
& \textbf{Trajectory Formulation$^{*}$} 
& \textbf{Optimization Objective} \\
\hline
\cite{Zhou2020} 
& Single UAV
& Transmitter
& TD
& Max-min secrecy rate  \\

\cite{Na2022} 
& Single UAV
& Relay
& TD
& Max-min secrecy rate \\

\cite{Wang2024} 
& Single UAV 
& Transmitter
& SHF
& Max-min secrecy throughput \\

\cite{Huang2025} 
& Single UAV 
& Transmitter
& APF
& Throughput$^{\dagger}$ \\
\hline
\cite{Miao2020} 
& Dual UAV 
& Relay \& Jammer
& TD
& Average secrecy rate \\

\cite{Cai2020} 
& Dual UAV 
& Transmitter \& Jammer
& TD
& Secrecy energy efficiency$^{\ddagger}$ \\

\cite{Lei2023} 
& Dual UAV
& Transmitter \& Jammer
& TD
& Max-min secrecy rate \\
\hline\hline
\textbf{This work} 
& \textbf{Dual UAV} 
& \textbf{Transmitter \& Jammer}
& \textbf{Co-SHF} 
& \textbf{Max-min secrecy throughput} \\
\hline
\end{tabular}

\vspace{0.8mm}
{\raggedright\footnotesize
\textit{Note:} 
$^{*}$ TD denotes time discretization. SHF/APF/Co-SHF are continuous-time trajectory formulations. 
$^{\dagger}$ Single-user covert communication: maximizing throughput under a low-probability-of-detection (LPD) constraint. 
$^{\ddagger}$ Predetermined jammer trajectory for tractability.\par}
\end{table*}

In this work, we study a dual-UAV transmitter-jammer system for multi-user secure communication, where the transmitter UAV serves multiple ground users and the jammer UAV cooperatively interferes with a ground eavesdropper. 
We aim to maximize the minimum secrecy throughput among users by jointly optimizing the continuous-time trajectories of both UAVs subject to mobility constraints, collision avoidance, and binary user scheduling. 
From a networked coordination perspective, the key challenge lies in the strong spatiotemporal coupling between heterogeneous UAV roles and user fairness, which makes the resulting continuous-time design problem highly nontrivial. 
To address this challenge, we reveal an exploitable coordination structure of the optimal dual-UAV motion in continuous time and leverage it to obtain a scalable solution without relying on dense time discretization.
The main contributions are summarized as follows: 
\begin{itemize}

\item \textbf{Characterization of optimal hovering coordination in dual-UAV networks:}
We rigorously characterize the optimal hovering behavior in a coupled dual-UAV system with heterogeneous roles and a common secrecy objective. Specifically, we prove that in an optimal solution the two UAVs must perform synchronized hovering at a finite set of co-hovering point pairs, whose total number is upper bounded by the number of ground users. This result reveals a fundamental coordination structure in which transmission and jamming are aligned in both time and space, and rules out suboptimal patterns where one UAV hovers while the other continues to fly.

\item \textbf{Characterization of optimal flying coordination between co-hovering locations:} We further characterize how the two UAVs coordinate their motion between co-hovering locations. We show that there exists an equivalent optimal solution in which, during each flight segment, at least one UAV travels at its maximum speed, while the other adjusts its speed to maintain coordination. This property specifies how the two UAVs should coordinate their motion during flying phases and eliminates inefficient wandering between hovering phases. By combining the hovering and flying characterizations above, we \textit{for the first time} characterize the optimal trajectory structure for dual-UAV coordination, referred to as the \textbf{collaborative SHF structure}.

\item \textbf{Collaborative SHF-based reformulation and efficient algorithm design:} 
Leveraging the proposed co-SHF structure, we reformulate the original continuous-time problem into an equivalent finite-dimensional optimization problem with a substantially reduced number of variables and constraints. Within this reformulation, continuous-time collision avoidance along each flight segment is enforced by a small set of minimum-distance inequalities. Building on this co-SHF-based representation, we develop an efficient iterative algorithm based on successive convex approximation, which converges to a stationary solution and scales favorably with the network size.

\item \textbf{Performance evaluation via numerical simulations:} 
Numerical results demonstrate that the proposed co-SHF-based design achieves secrecy performance comparable to, or even better than, high-precision TD-based schemes, while requiring only a small fraction of their runtime. Compared with no-jamming benchmarks, the proposed design consistently improves the minimum secrecy throughput, and its performance gains remain stable across different UAV speeds and jamming power levels, highlighting the benefits of structured dual-UAV coordination. 
\end{itemize}

The remainder of this paper is organized as follows. After reviewing related work in Section~\ref{sec:related work}, Section~\ref{sec:system model} presents the system model and formulates the optimization problem. Section~\ref{sec:characterization} characterizes the optimal coordinated trajectories for the dual-UAV transmitter-jammer system. Based on this characterization, Section~\ref{sec:problem reformulation} derives an equivalent finite-dimensional reformulation. Section~\ref{sec:iterative} develops an iterative algorithm based on successive convex approximation to solve the reformulated problem. Numerical results are provided in Section~\ref{sec:numerical result}, and Section~\ref{sec:conclusion} concludes the paper.

\section{Related Work}
\label{sec:related work}
This section reviews related studies on UAV-enabled secure communication, with emphasis on systems that exploit UAV mobility to enhance physical-layer security.
We first summarize UAV-enabled secure communication designs, and then focus on dual-UAV jamming-aided networks that employ coordinated aerial platforms.

\subsection{UAV-enabled Secure Communication}

UAV-enabled secure communication has been extensively investigated in single-UAV scenarios, where PLS is enhanced by exploiting the high controllability of UAV mobility.
A central research direction focuses on jointly optimizing the UAV trajectory and communication resources to improve secrecy performance under A2G channels.
Representative works study secrecy rate or secrecy throughput maximization by adapting the UAV flight path together with transmit power and communication scheduling under various network settings.
For instance, intelligent reflecting surface (IRS)-assisted UAV secure communication systems are considered in \cite{Pang2022,Li2021,Fang2021,Niu2022,Wang2023}, where the UAV trajectory is jointly designed with passive beamforming to enhance legitimate links while suppressing information leakage.
In addition, short-packet transmission and finite-blocklength effects are incorporated into single-UAV secure designs in \cite{Chen2023_2,Wang2020_3}, revealing the impact of latency and reliability constraints on secrecy-oriented trajectory optimization.
Beyond single-user secrecy enhancement, multi-user secure communication has also been studied in single-UAV networks by accounting for user scheduling, association, and fairness considerations. Several works investigate secrecy performance in multi-user scenarios by jointly optimizing the UAV trajectory and resource allocation, with objectives such as maximizing the minimum secrecy rate among users or balancing secrecy performance across the network \cite{Zhou2020,Na2022,Mao2023,Lu2024}.
These studies highlight the importance of incorporating fairness and scheduling decisions into UAV-enabled secure communication designs.
Furthermore, the single-UAV PLS framework has been extended to accommodate more sophisticated system models and constraints.
Examples include secure UAV communications with energy efficiency considerations \cite{Mao2023}, UAV-assisted relaying under secrecy constraints \cite{Lu2024,Tang2022,Niu2022}, non-orthogonal multiple access (NOMA)-based secure transmission \cite{Chen2020}, and fixed-wing or three-dimensional trajectory models \cite{Xiong2023}.
Collectively, these works demonstrate the versatility of single-UAV platforms for secure wireless communication across diverse network architectures.

To further enhance communication security, AN transmission and jamming-aided approaches have been incorporated into single-UAV secure communication designs to intentionally degrade eavesdroppers while preserving the quality of legitimate links. In this context, a variety of jamming-aided UAV-enabled secure communication networks have been investigated by jointly designing the UAV trajectory and interference-related parameters under different system settings  \cite{Ding2024,Li2023,Wang2020,Lu2022,Lu2022_2,Chen2023_3,Nnamani2020,Jia2024}.
For example, friendly jamming from ground nodes is considered in \cite{Wang2020_2}, where the UAV trajectory and transmission strategy are optimized to improve secrecy performance against passive eavesdroppers.
In \cite{Zhou2020_2}, jamming-aided secure communication is studied for uplink UAV-enabled networks, highlighting the impact of interference coordination on secrecy enhancement.
The integration of AN or jamming with more advanced system architectures has also attracted attention, such as IRS-assisted jamming-aided secure UAV communications \cite{Wen2024} and mobile edge computing (MEC)-enabled secure transmission with interference management \cite{Xu2021_2}.
In addition, friendly jamming generated by aerial platforms has been explored as a means to enhance security in single-UAV communication systems \cite{Zhou2022,Liu2024}.
These studies demonstrate that jamming and AN constitute effective security enhancement tools within single-UAV networks and naturally motivate architectures that explicitly separate information transmission and interference generation across dedicated aerial platforms, which are reviewed in the next subsection.

\subsection{Dual-UAV Jamming-aided Secure Communication}

Building upon jamming-aided secure network designs within single-UAV systems, dual-UAV transmitter-jammer architectures have been widely investigated to further enhance PLS by functionally separating information transmission and interference generation across two aerial platforms.
In such systems, one UAV serves as the information transmitter while the other acts as a dedicated jammer to intentionally degrade the eavesdropper’s reception, and the mobility of both UAVs provides additional spatial degrees of freedom for secrecy enhancement.
Representative works study dual-UAV jamming-aided secure communication by jointly optimizing the trajectories of the transmitter and jammer together with communication resource under various channel and network models \cite{Li2020,Li2021_2,Xu2021,Zhang2023,Kim2021,Kang2022,Diao2022,Ye2024}.
For example, \cite{Li2021_2} investigates joint resource, trajectory and AN-related designs under a dual-UAV setting. \cite{Xu2021,Zhang2023} further consider trajectory and resource allocation coupling to exploit the cooperation between the two aerial platforms.
In addition, the impact of friendly jamming UAVs on secure communications is examined in \cite{Kim2021}, while \cite{Kang2022} studies dual-UAV aided bi-directional communications, illustrating the broad applicability of the transmitter-jammer architectures.
Beyond the basic dual-UAV secrecy model, various performance objectives and system constraints have been considered to address more practical network requirements.
Robust dual-UAV jamming-aided designs have been developed to cope with imperfect channel state information, uncertain environments, or reliability requirements \cite{Wang2022_2,Ye2024,Pandey2025_2}.
Dual-UAV secure communication has also been investigated in spectrum-sharing or cognitive settings and relay-assisted architectures, where interference management and secrecy requirements are jointly accounted for \cite{Li2020,Wang2022_2}.
Moreover, security-oriented designs have been extended to incorporate additional system considerations, such as intrusion detection and energy-harvesting related mechanisms in UAV-enabled networks \cite{Vo2021}.

More recently, dual-UAV jamming-aided secure communication has been extended toward more networked and intelligent settings.
Some works generalize the transmitter-jammer architectures to multi-UAV or swarm-enabled secure networks, where multiple aerial platforms cooperate and the resulting coordination becomes more intricate \cite{Dang2022,Liu2023,Lei2025,Wang2021}.
Distributed and multi-role UAV coordination has also been explored from a network optimization perspective \cite{Zhong2024,Gao2024,Wen2025,Li2020_2}.
Furthermore, dual-UAV jamming has been applied to covert communication scenarios with LPD constraints, enriching the security paradigms supported by cooperative UAV networks \cite{Yang2023,Liu2024_2,Lei2024}.
Dual-UAV architectures have also been studied for secure IoT-oriented networks, including joint resource allocation and trajectory design \cite{Pandey2025,Pandey2025_2}.

Despite these advances, most existing works solve the problem by discretizing time and optimizing the UAV positions over a finite set of time slots. 
From a networked coordination perspective, a systematic characterization of the optimal coordinated continuous trajectory structure for dual-UAV transmitter-jammer systems, remains largely unexplored.

\begin{figure}[t]
	\centering
\includegraphics[width=8 cm,trim=10 25 10 10]{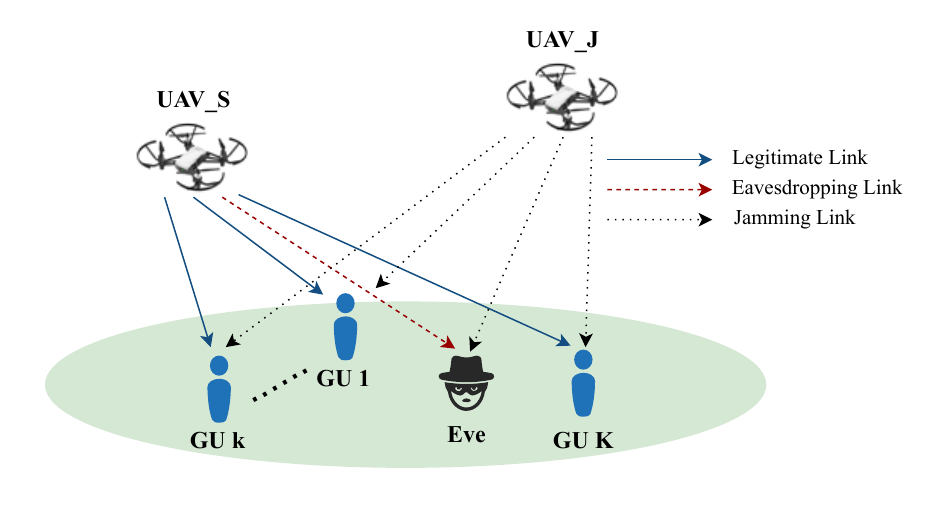}
\caption{An example of a dual-UAV-assisted wireless transmission scenario in the presence of an eavesdropper.\label{fig-sys}}
\end{figure} 

\section{System Model and Problem Formulation}
\label{sec:system model}

As depicted in Fig.~\ref{fig-sys}, we consider a UAV-assisted multi-user data transmission scenario where a UAV, referred to as UAV-S, is employed to provide communication service to $K$ ground users, while another UAV, designated as UAV-J, acts as a mobile jammer to emit interference to counter potential eavesdropping from Eve. 
Both UAVs are equipped with onboard sensing and RF localization modules that can leverage  the horizontal positions of the users and of the eavesdropper \cite{Bisio,Sallouha2025}. Without loss of generality, we adopt a 3D Cartesian coordinate system, where the horizontal locations of GU $k$ and Eve are denoted as ${{\bf{w}}_k} = ({w_{x,k}},{w_{y,k}}), \forall k \in {\cal K} \buildrel \Delta \over=\{ 1, \ldots, K\}$ and ${{\bf{w}}_e} = ({w_{x,e}},{w_{y,e}})$, respectively. In addition, both UAVs fly at a fixed altitude $H$ and their horizontal positions at any time $t \in [0, T]$ are denoted as ${\bf{q}}_S(t) = (x_S(t), y_S(t))$ and ${\bf{q}}_J(t) = (x_J(t), y_J(t))$, where $T$ is the total task time. In practice, the mobility of UAVs is subject to the following constraints
\begin{align}
    &{\bf q}_u(0)={\bf q}_{u,I},{\bf q}_u(T)={\bf q}_{u,F},u\in \{S,J\},\label{con:initial final point}\\
    &{\left\| {{\bf{\dot q}}_u(t)} \right\|^2} \le {V^2},u\in \{S,J\},\label{con:speed}    
\end{align}
where constraint \eqref{con:initial final point} requires UAVs to fly from fixed initial points ${\bf q}_{u,I}$ to fixed final points ${\bf q}_{u,F}$, and constraint \eqref{con:speed} represents that the flight speed of UAVs should be no more than the maximum allowable speed $V$. To avoid a collision between UAV-S and UAV-J, a minimum safety distance should be maintained between UAVs, leading to the following collision avoidance constraint
\begin{equation}
        {\left\| {{\bf q}_S(t)}-{{\bf q}_J(t)} \right\|^2} \ge {d^2_{min}}, \label{con:collision}
\end{equation}
where $d_{min}$ denotes the minimum distance between UAVs.

Considering the dominance of LoS links in rural areas with minimal blockage or in the urban macro scenario when the UAV flies above 100m \cite{3gpp,Muruganathan2021,Zeng2019}, we adopt the free-space path loss model to characterize the UAV-to-ground wireless channels.
Thus, the channel gain from UAVs to ground node $m\in\{k,e\}$ and Eve at time $t$ are respectively formulated as 
\begin{equation}
    g_{u,m}(t)=\frac{\beta_0}{d_{u,m}(t)^2}= \frac{{\beta _0}}{{{\left\| {{\bf{q}}_u(t) - {{\bf{w}}_m}} \right\|}^2} + {H^2}},
\end{equation}
where ${d_{u,m}}(t)$ represents the distance between UAV-$u$ and ground node $m\in\{k,e\}$, and $\beta_0$ refers to the channel power gain at the reference distance of unit meter.

Denote a set of binary variables ${a_k}(t) \in \{ 0,1\}$ as GU scheduling indicator. More specifically, GU $k$ can receive signals from UAVs at time $t$ if and only if $a_k(t)=1$, otherwise $a_k(t)=0$. In addition, UAVs are restricted to serve no more than one GU at any given time $t$. Thus, the communication scheduling constraints can be formulated as
\begin{align}
    &{a_k}(t) \in \{ 0,1\}, \forall k \in {{\cal K}},t \in [0,T],\label{con:scheduling1}\\
    &\sum\limits_{k = 1}^K {{a_k}(t) = 1}, \forall t \in [0,T].\label{con:scheduling2}
\end{align}

Accordingly, the transmitting rate of ground node $m\in\{k,e\}$ at time $t$ is formulated as 
\begin{equation}
    R_{S,m}(t)={\log _2}(1 + \frac{P_S{g_{S,m}}(t)}{P_J{g_{J,m}}(t)+\sigma_m^2}) ,
\end{equation}
where $P_S$ and $P_J$ denote the transmit power of UAV-S and UAV-J, and $\sigma_m^2$ is the additive white Gaussian noise (AWGN) power at ground node $m$. Thus, the secrecy rate of GU $k$ at time $t$ can be obtained as
\begin{equation}
        R_k({\bf{q}}_u(t))=\left[R_{S,k}(t)-R_{S,e}(t)\right]^+, 
\end{equation}
where ${\left[x\right]^ + } = \max \{ x,0\}$. 

Taking the fairness issue into account, our objective is to maximize the minimum secrecy throughput among GUs by jointly optimizing the trajectories of UAVs and communication scheduling while satisfying mobility and scheduling constraints. Consequently, the original problem can be formulated as
\begin{subequations}\label{OP}
    \begin{alignat}{2}
        ({\rm OP}):&\max \limits_{\{ {\bf{q}}_u(t),{a_k}(t)\} } &&\min \limits_{k \in {\cal K}} \int_0^T {a_k}(t)R_k({\bf{q}}_u(t))dt\\
       &\quad\quad \mathrm{s.t.:} &&\eqref{con:initial final point},\eqref{con:speed},\eqref{con:collision},\eqref{con:scheduling1},\eqref{con:scheduling2}. \nonumber
    \end{alignat}
\end{subequations}

Problem (OP) is difficult to solve optimally due to the following reasons. First, since UAV trajectories are continuous in both time and space, optimizing the trajectories ${\bf{q}}_u(t)$ and communication scheduling ${a_k}(t)$ involves an infinite number of variables, rendering (OP) intractable to solve. Second, the coupling between ${\bf{q}}_S(t)$ and ${\bf{q}}_J(t)$ poses further difficulties in handling the objective function $R_k({\bf{q}}_u(t))$. Moreover, the non-convex collision avoidance constraint \eqref{con:collision} and the integer constraint \eqref{con:scheduling1} add to the complexity of the problem. Therefore, how to find efficient solutions to (OP) remains challenging.

It should be noted that in prior work \cite{Miao2020,Cai2020}, researchers typically adopt a TD-based design to approximate continuous trajectories with a finite number of trajectory points. 
While this method offers a tractable dimension-reduction approach, it introduces substantial computational complexity, as the number of optimization variables grows rapidly with the discretization accuracy. Although decreasing the quantization accuracy may reduce computational complexity, it can also degrade trajectory resolution and lead to suboptimal solutions. To tackle these limitations, we characterize the optimal dual-UAV trajectory structure to solve (OP) efficiently.

\section{Characterization of Optimal Dual-UAV Coordinated Trajectories}
\label{sec:characterization}
In this section, we analyze the structural properties of the optimal trajectories in the dual-UAV cooperative jamming scenario defined by (OP). While the SHF structure has been proven optimal in single-UAV scenarios \cite{Yuan2021}, it is no longer applicable to the dual-UAV coordinated setting considered in (OP). 
In particular, the secrecy rate $R_k({\bf{q}}_u(t),{a_k}(t))$ depends jointly on the positions of both UAVs, leading to strong spatial coupling. Moreover, the optimal trajectory of each UAV is influenced by the movement of the other, creating time-dependent coordination requirements. Such spatial and temporal inter-dependencies prohibit the independent application of SHF to each UAV, necessitating a joint and coordinated trajectory design. 
To this end, we rigorously analyze the trajectory coupling in (OP) and develop a \textit{collaborative SHF} structure for dual-UAV coordination.

Specifically, we show that the optimal trajectories in this collaborative SHF structure, as shown in Fig.~\ref{fig-shf}, exhibit two fundamental properties: 
\begin{itemize}
    \item \textbf{Synchronized hovering with at most} $\boldsymbol{K}$ \textbf{co-hovering point pairs;}
    \item \textbf{Except during co-hovering, at least one UAV flies with maximum speed} $\boldsymbol{V}$.
\end{itemize}
We now rigorously prove the above two properties of the collaborative SHF structure. 
We begin by analyzing the coordination between the two UAVs during hovering and flight.  The following proposition shows that if hovering occurs, the UAVs must hover and fly simultaneously, and during flight, at least one UAV moves at the maximum speed.

\begin{figure}[t]
	\centering
\includegraphics[width=9 cm,trim=10 25 10 10]{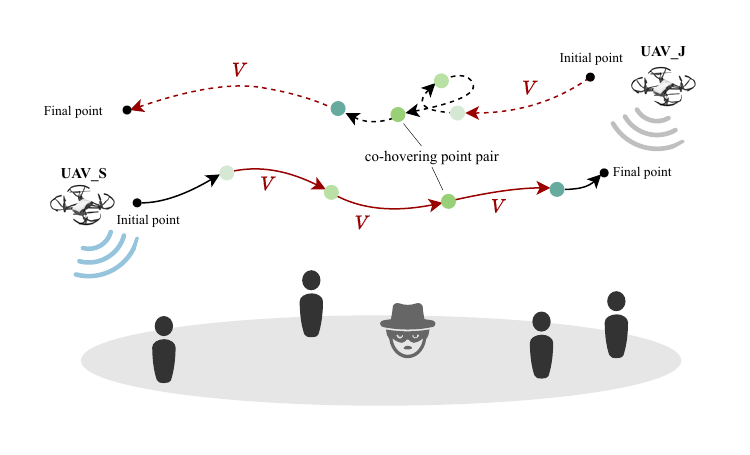}
\caption{Illustration of the collaborative SHF trajectory structure for dual-UAV cooperative jamming-assisted secure communication.\label{fig-shf}}
\end{figure} 


\begin{prop}
\label{prop1}
In the optimal solution $\{{\bf{q}}^*_u(t),{a^*_k}(t)\}$ to (OP), the two UAVs successively hover at a finite number of { synchronized co-hovering point pairs}, with at least one UAV traveling at the maximum speed $V$ during each flying segment between them.
\end{prop}

\begin{proof}
For contradiction, suppose that in the optimal trajectories $\{{\bf q}^*_u(t)\}$, there exists a sufficiently small time interval $[\tau_1, \tau_2]$ over which both UAVs fly at constant speeds strictly less than the maximum speed $V$, i.e., $0 < v_u < V$ for $u \in \{S, J\}$. The interval is short enough such that the UAV trajectories can be locally approximated as straight-line segments, with constant speeds and fixed communication scheduling. Specifically, we assume that only user $k_0$ is scheduled, i.e., $a_{k_0}(t) = 1$ for all $t \in [\tau_1, \tau_2]$. Let ${\bf q}_{u,1}$ and ${\bf q}_{u,2}$ denote the positions of UAV-$u$ at times $\tau_1$ and $\tau_2$, respectively. Then the position of UAV-$u$ at time $t \in [\tau_1, \tau_2]$ is given by  
${\bf q}_u(t) = {\bf q}_{u,1} + \frac{t - \tau_1}{\tau_2 - \tau_1}({\bf q}_{u,2} - {\bf q}_{u,1})$. Accordingly, the secrecy throughput of user $k_0$ in this interval is $\Delta U_{k_0} = \int_{\tau_1}^{\tau_2} R_{k_0}({\bf q}_S(t), {\bf q}_J(t))\,dt$.



We then establish that accelerating the UAV with the longer segment between ${\bf q}_{u,1}$ and ${\bf q}_{u,2}$ to speed $V$, and reallocating the resulting saved time to coordinated hovering yields a higher secrecy throughput, contradicting the assumed optimality.

\begin{figure*}[t!]
\begin{align}
\Delta U_{k_0} &= \frac{v_f}{V} \int_{\tau_1}^{\tau_2} R_{k_0}({\bf q}_S^*(t), {\bf q}_J^*(t))\, dt + \frac{V - v_f}{V} \int_{\tau_1}^{\tau_2} R_{k_0}({\bf q}_S^*(t), {\bf q}_J^*(t))\, dt \nonumber\\
&\hspace{-.65cm}\overset{v_ft = Vt_1}{\overset{v_ft = \frac{Vv_f}{V - v_f} t_2}{=}} \underbrace{\int_0^{\frac{\|{\bf q}_{f,2} - {\bf q}_{f,1}\|}{V}} R_{k_0}({\bf q}_S^*(t_1), {\bf q}_J^*(t_1))\, dt_1}_{\Delta \bar{U}_{k_0}} + \underbrace{\int_0^{\frac{\|{\bf q}_{f,2} - {\bf q}_{f,1}\|}{V - v_f}} R_{k_0}({\bf q}_S^*(t_2), {\bf q}_J^*(t_2))\, dt_2}_{\Delta \hat{U}_{k_0}}.\label{eqvi}
\end{align}
\hrulefill
\end{figure*}

To facilitate the analysis, we decompose $\Delta U_{k_0}$ into two components, $\Delta \bar{U}_{k_0}$ and $\Delta \hat{U}_{k_0}$, as shown in \eqref{eqvi}, corresponding to two decomposed flight segments based on the acceleration of one UAV to speed $V$. Let $f \in \{S, J\}$ denote the UAV with the longer segment between ${\bf q}_{u,1}$ and ${\bf q}_{u,2}$. In the case where both UAVs travel the same distance, either UAV can be designated as $f$ without loss of generality. The two terms in \eqref{eqvi} are interpreted as follows:

\begin{itemize}
\item \textbf{For $\Delta \bar{U}_{k_0}$}: UAV-$f$ travels from ${\bf q}_{f,1}$ to ${\bf q}_{f,2}$ at the maximum speed $V$, completing the segment in duration $\Delta \bar{\tau} = \frac{\|{\bf q}_{f,2} - {\bf q}_{f,1}\|}{V}$. The other UAV, denoted by $g \in \{S, J\} \setminus \{f\}$, follows a straight path from ${\bf q}_{g,1}$ to ${\bf q}_{g,2}$ at an adjusted speed of $\alpha V$, where $\alpha = \frac{v_g}{v_f}$ ensures synchronized arrival at their respective endpoints.

\item \textbf{For $\Delta \hat{U}_{k_0}$}: UAV-$f$ moves along the same segment at a constant speed of $\frac{Vv_f}{V - v_f}$, corresponding to a stretched duration of $\Delta \hat{\tau} = \frac{V - v_f}{Vv_f} \|{\bf q}_{f,2} - {\bf q}_{f,1}\|$. Meanwhile, UAV-$g$ travels linearly between ${\bf q}_{g,1}$ and ${\bf q}_{g,2}$ at speed $\frac{Vv_g}{V - v_f}$, preserving the same scaling ratio $\alpha$.
\end{itemize}

It is worth noting that in the special case where $v_f = v_g$, the scaling ratio becomes $\alpha = 1$. In this case, either UAV can be selected as $f$, and both decomposed integrals in \eqref{eqvi} remain well-defined: $\Delta \bar{U}_{k_0}$ corresponds to both UAVs flying at speed $V$, while $\Delta \hat{U}_{k_0}$ corresponds to both UAVs flying at an equal speed of $\frac{Vv_f}{V - v_f}$. The decomposition structure remains valid and comparable.
By decomposing $\Delta U_{k_0}$ in this way, we can interpret the original trajectory segment as being equivalent, in terms of throughput, to the sum of two traversals from ${\bf q}_{f,1}$ to ${\bf q}_{f,2}$: one where UAV-$f$ flies at the maximum speed $V$, and another at the unconstrained speed $\frac{Vv_f}{V - v_f}$. In both traversals, the other UAV-$g$ moves along its corresponding path segment at adjusted speeds $\alpha V$ and $\frac{Vv_g}{V - v_f}$, respectively, such that the two UAVs arrive at their endpoints simultaneously in each case. Note that the total duration of the two traversals remains equal to the original interval length, i.e., $\tau_2 - \tau_1$.

Let $R_{k_0}^{\text{max}}$ denote the maximum instantaneous secrecy rate attained along the optimal trajectories during the interval $[\tau_1, \tau_2]$, i.e., $R_{k_0}^{\text{max}} := \max_{t \in [\tau_1, \tau_2]} R_{k_0}({\bf q}_S(t), {\bf q}_J(t))$. Since the instantaneous secrecy rate $R_{k_0}(t)$ generally varies over time during the interval $[\tau_1, \tau_2]$, the integral throughput $\Delta \hat{U}_{k_0}$ is strictly less than the upper bound $\Delta \hat{\tau} \cdot R_{k_0}^{\text{max}}$, i.e., $\Delta \hat{U}_{k_0} < \Delta \hat{\tau} \cdot R_{k_0}^{\text{max}}$. 

We now construct a new trajectory by reallocating the surplus duration $\Delta \hat{\tau}$ to hovering. To preserve the temporal alignment and structural consistency established in the decomposition \eqref{eqvi}, the two UAVs are required to perform the hovering behavior in a synchronized manner. 
Specifically, both UAVs remain stationary at a pair of locations $({\bf q}_{S,k_0}^{\text{max}}, {\bf q}_{J,k_0}^{\text{max}})$ for duration $\Delta \hat{\tau}$, where the secrecy rate achieves $R_{k_0}^{\text{max}}$. 

By allowing both UAVs to hover at these positions for duration $\Delta \hat{\tau}$, the new throughput becomes
\begin{equation}
\Delta U_{k_0}' =\Delta \bar{U}_{k_0} + \underbrace{ R_{k_0}^{\text{max}}({\bf q}_{S,k_0}^{\text{max}}, {\bf q}_{J,k_0}^{\text{max}}){\Delta \hat{\tau}}}_{\Delta U'_{k_0}} ,
\end{equation}
Since hovering at the optimal positions yields a higher instantaneous throughput than moving along the original optimal paths, we have $\Delta U_{k_0}' > \Delta U_{k_0}=\Delta \bar{U}_{k_0} + \Delta {\hat{U}}_{k_0}$. 

Therefore, a new trajectory can be constructed within this region, where both UAVs first fly from ${\bf q}_{u,1}$ to a selected co-hovering point pair $({\bf q}_{S,k_0}^{\text{max}}, {\bf q}_{J,k_0}^{\text{max}})$ at adjusted speeds—UAV-$f$ at the maximum speed $V$ and UAV-$g$ at a scaled speed to ensure synchronized arrival. They then remain stationary at the co-hovering point pair for a duration of $\Delta \hat{\tau}$, followed by continuing to ${\bf q}_{u,2}$ along straight-line paths, again at appropriately adjusted speeds to reach the destination simultaneously. The mission time of this new trajectory remains identical to that of the original optimal trajectory ${\bf q}^*_u(t)$, while achieving a higher secrecy rate. Notably, the newly constructed trajectory preserves the topological continuity of the trajectory points and leaves the trajectory outside this region unaffected. Therefore, the new trajectory surpasses the so-called optimal trajectory in performance. This contradicts the assumption that the original trajectories are optimal.

\end{proof}

\vspace{-.2cm}
According to Proposition~\ref{prop1}, the UAVs successively hover at a finite number of co-hovering point pairs, and at least one UAV (the one covering the longer segment) travels at the maximum speed $V$ during each flight between these pairs. To further reduce the structural complexity and the number of optimization variables, we analyze the total number of such hovering pairs in Proposition~\ref{prop2}.

\begin{prop}
\label{prop2}
In (at least one equivalent) optimal solution $\{{\bf{q}}^*_u(t), {a^*_k}(t)\}$ to problem (OP), there are at most $K$ { co-hovering point pairs for UAVs}, where $K$ is the number of GUs. 
\end{prop}

\begin{proof}
We prove the proposition by contradiction. Assume that in the optimal solution $\{{\bf{q}}^*_u(t), {a^*_k}(t)\}$, there exists a GU $k$ assigned more than one hovering point, i.e., $N > 1$. Let these hovering points be associated with durations $(\tau_{1,k}, \tau_{2,k}, \ldots, \tau_{N,k})$ and corresponding transmission rates $(r_{1,k}, r_{2,k}, \ldots, r_{N,k})$.

For the trajectory to be optimal, the transmission rates at all hovering points should be consistent. If there is inconsistency, let $r_{k,\text{max}}$ denote the maximum transmission rate among these points. We can construct a new trajectory by redistributing the hovering times. Specifically, the total hovering time for GU $k$ can be defined as $\tau'_k = \frac{\sum_{i=1}^N r_{i,k} \tau_{i,k}}{r_{k,\text{max}}}$. This total time $\tau'_k$ can be allocated entirely to the hovering point with transmission rate $r_{k,\text{max}}$, while setting the hovering times at other points to zero. This adjustment maintains the total time cost of the trajectory but increases the throughput, as $r_{k,\text{max}} \cdot \tau'_k > \sum_{i=1}^N r_{i,k} \tau_{i,k}$. Thus, a new trajectory with better performance can be constructed, contradicting the optimality of the original trajectory. Since an equivalent trajectory with the same performance as the original can be constructed by assigning the total hovering time $\sum_{i=1}^N \tau_{i,k}$ to any one of the hovering points, it follows that the optimal solution must allocate at most one hovering point to each ground user. Therefore, in at least one equivalent optimal solution to problem (OP), there can be at most $K$ hovering points, where $K$ is the number of GUs.
\end{proof}

So far, we have characterized the structure of the optimal dual-UAV trajectories under the considered cooperative jamming scenario. The optimal trajectories follow a collaborative SHF structure comprising at most $K$ synchronized co-hovering point pairs. Leveraging this structure, we proceed to reformulate problem (OP) with limited optimization variables in the next section. 

\section{Problem Reformulation}
\label{sec:problem reformulation}

According to the trajectory structure characterized in Section~\ref{sec:characterization}, which follows the collaborative SHF model, the two UAVs successively visit at most $K$ synchronized co-hovering point pairs, each associated with a finite hovering duration. Let ${\bf q}_{u,i} = (x_{u,i}, y_{u,i})$ denote the locations of the $i$-th hovering point for UAV $u \in \{S, J\}$, and ${\bf t} = (t_1, \ldots, t_K)$ denote the corresponding hovering durations, where $t_i$ is the hovering duration spent at the $i$-th co-hovering pair.
In the context of 2D trajectory design, both UAVs may follow curved paths between adjacent co-hovering point pairs. To approximate such curves, we introduce $N$ turning points between adjacent hovering points for each UAV. Specifically, let ${\bf q}_{u,i,j} = (x_{u,i,j}, y_{u,i,j})$ denote the $j$-th turning point for UAV $u$ between its $i$-th and $(i+1)$-th hovering positions, where $j = 1, \ldots, N$ and $i = 1, \ldots, K$. 
Thus, the entire trajectory of UAV-$u$, for $u\in\{S,J\}$, including both hovering points and turning points in sequential order is expressed as 
\begin{equation}
        {\bf{Q}}_u \buildrel \Delta \over = ({{\bf{q}}_{u,0,0}},...,{{\bf{q}}_{u,i,j}},...,{{\bf{q}}_{u,K,N}},{{\bf{q}}_{u,K + 1,0}}),
\end{equation}
where ${{\bf{q}}_{u,0,0}}$ and ${{\bf{q}}_{u,K+1,0}}$ refer to the initial and final points of UAV-$u$, respectively. Hence, the constraint \eqref{con:initial final point} can be rewritten to 
\begin{equation}
\begin{aligned}
    {{\bf{q}}_{u,0,0}}= {{\bf{q}}_{u,I}},{{\bf{q}}_{u,K+1,0}} = {{\bf{q}}_{u,F}}.\label{con:initial final point2}
\end{aligned}    
\end{equation}

Next, we analyze the flight segments between hovering points along the UAV trajectories. For notational clarity, let $d_{i,j}({\bf Q}_u) = \| {\bf q}_{u,i,j+1} - {\bf q}_{u,i,j} \|$ denote the distance between two adjacent trajectory points of UAV $u$. The flight time of each segment $(i,j)$ is determined by the UAV with the longer travel distance, which flies at the maximum speed $V$, while the other UAV adjusts its speed proportionally to ensure simultaneous arrival. Accordingly, the flight duration is given by $\Delta t_{i,j} = \frac{\max\{ d_{i,j}({\bf Q}_S),\ d_{i,j}({\bf Q}_J) \}}{V}$, and the speed of UAV-$u$ is $\frac{d_{i,j}({\bf Q}_u)}{\Delta t_{i,j}}$. Let $\tau_{i,j} \in [0,\Delta t_{i,j}]$ denote the elapsed time since departing from the $j$-th turning point between the $i$-th and $(i{+}1)$-th hovering points. The position of UAV-$u$ at time $\tau_{i,j}$ along this segment is given by
\begin{equation}
\begin{aligned}
  {\bf{\bar q}}_{u,i,j}(\tau_{i,j}) \!=\! {\bf q}_{u,i,j} \!+\! \frac{({\bf q}_{u,i,j+1} \!-\! {\bf q}_{u,i,j}) \!\cdot\! \tau_{i,j} \!\cdot\! d_{i,j}({\bf Q}_u)}{\| {\bf q}_{u,i,j+1} \!-\! {\bf q}_{u,i,j} \| \!\cdot\! \Delta t_{i,j}}.
\end{aligned}    
\end{equation}

Hence, the total task time, consisting of all hovering durations and flight durations, must satisfy
\begin{equation}
    \sum\limits_{i = 1}^K t_i + \sum\limits_{i = 0}^K \sum\limits_{j = 0}^N \frac{\max \{ d_{i,j}({\bf Q}_S),\ d_{i,j}({\bf Q}_J) \}}{V} \le T. \label{con:task time}
\end{equation}

To provide a unified representation of UAV trajectories across flight segments, we introduce a normalized time variable. Specifically, the position of UAV-$u$ within the $(i,j)$-th flight segment is expressed as ${\bf q}_{u,i,j}(z) = {\bf q}_{u,i,j} + z \cdot ({\bf q}_{u,i,j+1} - {\bf q}_{u,i,j})$, where $z \in [0,1]$ denotes the normalized variable along the segment from ${\bf q}_{u,i,j}$ to ${\bf q}_{u,i,j+1}$. Thus, the collision avoidance constraint between UAV-$S$ and UAV-$J$ can be formulated as
\begin{equation}
    \left\| {\bf q}_{S,i,j}(z) - {\bf q}_{J,i,j}(z) \right\|^2 \ge d_{\min}^2, \forall i \in \mathcal{K},\ j \in \mathcal{N}, \label{con:collision2}
\end{equation}
which leads to an infinite number of constraints over continuous time, thereby significantly increasing the complexity of the optimization problem.

Since the optimal scheduling strategy divides the UAV trajectory into multiple continuous segments, each corresponding to a specific operational region and assigned to a single user \cite{Yuan2023}, we define the scheduling of the $k$-th GU at the $i$-th hovering point as $a_{i,k}$, while the scheduling during the flight segment from ${{\bf{q}}_{u,i,j}}$ to ${{\bf{q}}_{u,i,j+1}}$ is represented by $a_{i,j,k}$. Consequently, the communication scheduling set ${\bf a}$ consists of both hovering and flight scheduling variables for each GU, and satisfy the following modified constraints
\begin{align}
      &{\bf a} \in \{ 0,1\},\label{con:scheduling1_2}\\
      &\sum\limits_{k = 1}^K {{a_{i,k}} = 1},\sum\limits_{k = 1}^K {{a_{i,j,k}} = 1}.\label{con:scheduling2_2}   
\end{align}
With the variables reformed above, the throughput of the $k$-th GU over the entire flight can be rephrased as
\begin{align}
    &{U_k}({\bf{Q}}_u,{\bf{t}},{\bf{a}}) \nonumber \\
    &= \sum\limits_{i = 1}^K {{a_{i,k}}{t_i}R_k({\bf{q}}_{u,i})}   \\
    &\quad + \sum\limits_{i = 0}^K {\sum\limits_{j = 0}^N {a_{i,j,k}\Delta t_{i,j}\int_0^1{R_k({{\bf{q}}_{u,i,j}}(z))dz} }},\nonumber
\end{align}

It is worth noting that increasing the number of turning points leads to a more accurate approximation of the UAV flight paths, but also incurs higher computational complexity.  Therefore, selecting the number of turning points involves a trade-off between performance and computational complexity. 

After all, Problem (OP) is transformed into
\begin{subequations}\label{P2}
    \begin{alignat}{2}
    ({\rm P1}):&\max \limits_{{\bf{Q}}_u,{\bf t},{\bf a} } &&\min \limits_{k \in {\cal K}} U_k({\bf{Q}}_u,{\bf t},{\bf a})\\
       &\quad \mathrm{s.t.: } \quad &&  {\bf t} \ge 0, \label{con:duration}\\
       &&&\eqref{con:initial final point2},\eqref{con:task time},\eqref{con:collision2},\eqref{con:scheduling1_2},\eqref{con:scheduling2_2}.\nonumber
    \end{alignat}
\end{subequations} 
where the constraint \eqref{con:duration} ensures that the hovering duration is non-negative, while also allowing for the possibility of zero hovering duration. This is in accordance with Proposition~\ref{prop2}, which states that the number of co-hovering point pairs in the optimal trajectory is at most $K$, thus accommodating the case where fewer than $K$ co-hovering point pairs are involved without affecting the optimality of the solution. 
It can be observed that the constraint \eqref{con:scheduling1_2} is an integer constraint, which introduces discontinuity in the solution space. To address this, we relax the constraint \eqref{con:scheduling1_2} and obtain the problem (P2) as follows 
\begin{subequations}\label{P3}
    \begin{alignat}{2}
       ({\rm P2}):&\max \limits_{{\bf{Q}}_u,{\bf t},{\bf a} } &&\min \limits_{k \in {\cal K}} U_k({\bf{Q}}_u,{\bf t},{\bf a})\\
       &\quad \mathrm{s.t.: }\quad && 0\le {\bf a} \le 1, \label{con:scheduling1_3} \\
       &&&\eqref{con:initial final point2},\eqref{con:task time},\eqref{con:collision2},\eqref{con:scheduling2_2},\eqref{con:duration}.\nonumber
    \end{alignat}
\end{subequations}
Despite the dimensionality reduction of optimization variables enabled by the collaborative SHF structure, problem (P2) remains non-convex due to two main factors: the infinite number anti-collision constraints \eqref{con:collision2} arising from continuous-time trajectory during flight periods, and the non-concavity of the objective function $U_k({\bf{q}}_u(t),{a_k}(t))$. In the next section, we develop a tractable reformulation for the anti-collision constraint via conservative minimum-distance approximation, and propose a successive convex approximation algorithm to efficiently solve the reformed problem.

\section{Iterative Solution}
\label{sec:iterative}
In this section, we first reformulate the anti-collision constraint into a finite number of tractable constraints using a conservative minimum-distance approximation. Then, we construct concave lower bounds for both the distance terms and the secrecy throughput. Based on these approximations, we develop an efficient iterative solution to solve (P2).

\subsection{Convex Approximation for Anti-Collision Constraints}
For each pair of UAV flight segments indexed by $(i,j)$, constraint~\eqref{con:collision2} requires the squared distance $\|{\bf q}_{S,i,j}(z) - {\bf q}_{J,i,j}(z)\|^2$ to remain above the safety threshold $d^2_{\min}$ for all $z \in [0,1]$, which results in an infinite number of non-convex constraints over a continuous interval.

To tackle this issue, we first derive a convex lower-bound approximation for the distance expression. Noting that the squared norm $\|{\bf q}_{S,i,j}(z) - {\bf q}_{J,i,j}(z)\|^2$ is jointly convex with respect to both UAV trajectories, we apply first-order Taylor expansion at the local point $({\bf q}^{(r)}_{S,i,j}(z), {\bf q}^{(r)}_{J,i,j}(z))$, leading to the following inequality as
\begin{equation}
\begin{aligned}
& \|{\bf q}_{S,i,j}(z) - {\bf q}_{J,i,j}(z)\|^2 \\
&\ge -\|{\bf q}^{(r)}_{S,i,j}(z) - {\bf q}^{(r)}_{J,i,j}(z)\|^2 \\
&\quad + 2({\bf q}^{(r)}_{S,i,j}(z) - {\bf q}^{(r)}_{J,i,j}(z))^T ({\bf q}_{S,i,j}(z) - {\bf q}_{J,i,j}(z)) \\
&\triangleq D^{(r)}_{i,j}(z), \label{con:collision3}
\end{aligned}
\end{equation}

By substituting the original distance constraint with this linear lower bound, we obtain a convex approximation as
\begin{equation}
D^{(r)}_{i,j}(z) \ge d^2_{\min}. \label{con:collision4}
\end{equation}

However, this constraint still involves infinitely many values of $z \in [0,1]$. To further reduce complexity, we equivalently reformulate~\eqref{con:collision4} by enforcing its minimum over the interval to exceed the threshold and can obtain
\begin{equation}
\min_{z \in [0,1]} D^{(r)}_{i,j}(z) \ge d^2_{\min}. \label{con:collision5}
\end{equation}

Since $D^{(r)}_{i,j}(z)$ is a quadratic function over $z \in [0,1]$, its minimum can be efficiently obtained in closed form \cite{Wu2024}. Therefore, $\min_{z \in [0,1]} D^{(r)}_{i,j}(z)$ can be obtained as 
\begin{equation}
    \eta^{(r)}_{i,j} = \begin{cases}
D^{(r)}_{i,j}(\frac{-\theta_{i,j}}{2\zeta_{i,j}}), \zeta_{i,j}>0\& \frac{-\theta_{i,j}}{2\zeta_{i,j}} \in (0,1), \\
\min\{D^{(r)}_{i,j}(0),D^{(r)}_{i,j}(1)\}, \text{others} .
\end{cases}\label{collision-min}
\end{equation}
where to simplify the expression, we denote $\Delta {\bf q}_{i,j}={\bf q}_{S,i,j}-{\bf q}_{J,i,j}$ and $\theta_{i,j}$ and $\zeta_{i,j}$ are expressed as 
\begin{subequations}
\begin{align}
    \theta_{i,j} =& -(\Delta {\bf q}^{(r)}_{i,j+1}-\Delta {\bf q}^{(r)}_{i,j})^2 \nonumber \\
    &+ 2 (\Delta {\bf q}^{(r)}_{i,j+1}-\Delta {\bf q}^{(r)}_{i,j})(\Delta {\bf q}_{i,j+1}-\Delta {\bf q}_{i,j}), \\
    \zeta_{i,j} =& -2 \Delta {\bf q}^{(r)T}_{i,j} (\Delta {\bf q}^{(r)}_{i,j+1}-\Delta {\bf q}^{(r)}_{i,j}-\Delta {\bf q}_{i,j+1}+\Delta {\bf q}_{i,j}) \nonumber \\
    &+2 (\Delta {\bf q}^{(r)}_{i,j+1}-\Delta {\bf q}^{(r)}_{i,j})^T \Delta {\bf q}_{i,j}.     
\end{align}
\end{subequations}

As a result, the infinite set of convex constraints in~\eqref{con:collision4} can be equivalently replaced by the following form
\begin{equation}
\eta^{(r)}_{i,j} \ge d^2_{\min}, \label{con:collision6}
\end{equation}
which ensures that the minimum safe distance is guaranteed throughout each flight segment.

\subsection{Convex Approximation for Objective Function}
To obtain a concave approximation of $U_k({\bf{Q}}_u,{\bf t},{\bf a})$, we construct a lower-bound concave function $U_k^{(r)}({\bf{Q}}_u,{\bf t},{\bf a})$ which meets $\min \limits_{k \in {\mathcal K}} U_k({\bf{Q}}_u,{\bf t},{\bf a})\ge U_k^{(r)}({\bf{Q}}_u,{\bf t},{\bf a})$ and the equality holds at the local point $({\bf{Q}}_u^{(r)},{\bf{t}}^{(r)},{\bf{a}}^{(r)})$. 

\textit{1) Hovering Period:} Given the non-smooth operator ${\left[\cdot \right]^ + }$ makes the problem hard to handle, we first introduce a judgment matrix $J_{1,i,k}^{(r)}$ in each iteration and $J_{1,i,k}^{(r)}$ is defined~as 
\begin{equation}
    J_{1,i,k}^{(r)} = \begin{cases}
0, & \text{when  } R^{(r)}_{S,k}({\bf{q}}^{(r)}_{u,i})-R^{(r)}_{S,e}({\bf{q}}^{(r)}_{u,i}) < 0, \\
1, & \text{when  }  R^{(r)}_{S,k}({\bf{q}}^{(r)}_{u,i})-R^{(r)}_{S,e}({\bf{q}}^{(r)}_{u,i}) \ge 0.
\end{cases}
\end{equation}
which ensures $R_k^{(r)}({\bf{q}}^{(r)}_{u,i}) \ge J_{1,i,k}^{(r)}(R^{(r)}_{S,k}({\bf{q}}^{(r)}_{u,i})-R^{(r)}_{S,e}({\bf{q}}^{(r)}_{u,i})$. 

Let $d_k({\bf{q}}_{u,i})$ and $d_e({\bf{q}}_{u,i})$ represent the distances from UAV-$u$ to GU $k$ and Eve at the $i$-th co-hovering point pair, respectively. Therefore, the received SNRs at ground node $m\in\{k,e\}$ at the $i$-th co-hovering point pair can be expressed as
\begin{equation}
    {\gamma_{i,m}} =\frac{\beta_0P_S d^2_m({\bf{q}}_{J,i})}{\beta_0P_J d^2_m({\bf{q}}_{S,i})+\sigma_m^2 d^2_m({\bf{q}}_{S,i}) d^2_m({\bf{q}}_{J,i})}.
\end{equation}

As the functions $f(u)=-\log_2(1+u)$ and $g(v)=\log_2(1+v)$ hold convexity for $u$ and $\frac{1}{v}$ respectively, we have
\begin{equation}
\begin{aligned}
    & R_k({\bf{Q}}_u,{\bf t},{\bf a}) \\
    &\ge -J_{1,i,k}^{(r)}(\frac{A_{1,i,k}^{(r)}}{{\gamma_{i,k}}} + A_{2,i,k}^{(r)}{\gamma_{i,e}} )+ B_{1,i,k}^{(r)}J_{1,i,k}^{(r)},\label{approx:hov}   
\end{aligned}
\end{equation}
where $A_{1,i,k}^{(r)}=\frac{({\gamma^{(r)}_{i,k}})^2}{\ln2({\gamma^{(r)}_{i,k}}+1)}$, $A_{2,i,k}^{(r)}=\frac{1}{\ln2({\gamma^{(r)}_{i,e}}+1)}$, and $B_{1,i,k}^{(r)}=\frac{A_{1,i,k}^{(r)}}{{\gamma^{(r)}_{i,k}}} + A_{2,i,k}^{(r)}{\gamma^{(r)}_{i,e}}+U_k({\bf{q}}^{(r)}_{u,i})$ are positive constants. 

Since the expressions $\frac{1}{{\gamma_{i,k}}}$ and ${\gamma_{i,e}}$ involve multiplicative coupling among trajectory-related terms, we apply the arithmetic mean–geometric mean inequality to facilitate decoupling. Specifically, for $M$ non-negative variables $x_i$, the inequality $\sqrt[M]{{\prod\nolimits_{i = 1}^M {{x_i}} }} \le \frac{1}{M}\sum\nolimits_{i = 1}^M {{x_i}}$ holds, with equality if and only if all $x_i$ are equal. This allows us to decouple the terms and obtain
\begin{subequations}
\begin{align}
&\frac{1}{{\gamma_{i,k}}}\le \frac{P_J}{P_S}(\frac{F_{1,i,k}^{(r)}}{2 d^4_k({\bf{q}}_{J,i})}+\frac{d^4_k({\bf{q}}_{S,i})}{2F_{1,i,k}^{(r)}}),\\
&{\gamma_{i,e}} \le E_{1,i,k}^{(r)}(\frac{F_{2,i,k}^{(r)}}{2 d^4_e({\bf{q}}_{S,i})}+\frac{d^4_e({\bf{q}}_{J,i})}{2F_{2,i,k}^{(r)}})+\frac{E_{2,i,k}^{(r)}}{d^2_e({\bf{q}}_{S,i})},
\end{align}
\end{subequations}
where $F_{1,i,k}^{(r)}=d_k^2({\bf{q}}^{(r)}_{S,i}) d_k^2({\bf{q}}^{(r)}_{J,i})$, $F_{2,i,k}^{(r)} = d_e^2({\bf{q}}^{(r)}_{S,i}) d_e^2({\bf{q}}^{(r)}_{J,i})$, $E_{1,i,k}^{(r)} = \frac{\beta_0^2 P_S P_J}{\left( \beta_0 P_J + \sigma_e^2 d_e^2({\bf{q}}_{J,i}^{(r)}) \right)^2}$, and $E_{2,i,k}^{(r)} = -E_{1,i,k}^{(r)} d_e^2({\bf{q}}_{J,i}^{(r)}) + \frac{\beta_0 P_S d_e^2({\bf{q}}_{J,i}^{(r)})}{\beta_0 P_J + \sigma_e^2 d_e^2({\bf{q}}_{J,i}^{(r)})}$ are positive constants. 

Note that in the approximated expressions of $\frac{1}{\gamma_{i,k}}$ and $\gamma_{i,e}$, the terms $\frac{1}{d_k^4({\bf q}_{J,i})}$, $\frac{1}{d_e^4({\bf q}_{S,i})}$, and $\frac{1}{d_e^2({\bf q}_{S,i})}$ are still non-convex due to inverse distance terms. Since the function $f(x)=\frac{1}{x},x>0$ is monotonically decreasing and convex in $x$, we can construct convex upper bounds by applying concave lower bounds to $d_k^2({\bf q}_{J,i})$ and $d_e^2({\bf q}_{S,i})$, which are convex in their respective variables. Specifically, we have the following first-order Taylor approximations 
\begin{subequations}
\begin{align}
d_m^2({\bf q}_{J,i}) 
&\ge 2({\bf q}_{J,i}^{(r)} - {\bf w}_m)^T({\bf q}_{J,i} - {\bf w}_m) - \left\| {\bf q}_{J,i}^{(r)} - {\bf w}_m \right\|^2 \nonumber\\
   & \triangleq (d_m^{(r)}({\bf{q}}_{J,i}))^2, m\in\{k,e\},
\end{align}
\end{subequations}
with equality holding at ${\bf q}_{J,i} = {\bf q}_{J,i}^{(r)}$ and ${\bf q}_{S,i} = {\bf q}_{S,i}^{(r)}$, respectively. To ensure the convexity and tight upper bound of the above approximation, we impose the following linear constraints
\begin{equation}
    2({\bf q}_{J,i}^{(r)} - {\bf w}_m)^T({\bf q}_{J,i} - {\bf w}_m) - \left\| {\bf q}_{J,i}^{(r)} - {\bf w}_m \right\|^2 \geq 0.
\end{equation}

Combining the above approximations, the non-convex term $R_k({\bf{Q}}_u,{\bf t},{\bf a})$ is lower bounded as
\begin{equation}
\begin{aligned}
     &R_k({\bf{Q}}_u,{\bf t},{\bf a})\\
     &\ge -\Big(A_{1,i,k}^{(r)} \frac{\sigma_k^2 d_k^2({\bf q}_{S,i})}{\beta_0 P_S}\!+\! \frac{P_J A_{1,i,k}^{(r)}}{P_S}\big(\frac{F_{1,i,k}^{(r)}}{2(d_k^{(r)}({\bf{q}}_{J,i}))^4}\\
     &\quad+\!\frac{d^4_k({\bf{q}}_{S,i})}{2F_{1,i,k}^{(r)}}\big)\!+\!A_{2,i,k}^{(r)}E_{1,i,k}^{(r)}\big(\frac{F_{2,i,k}^{(r)}}{2(d_e^{(r)}({\bf{q}}_{S,i}))^4} \!+\!\frac{d^4_e({\bf{q}}_{S,i})}{2F_{2,i,k}^{(r)}}\big)\\
     &\quad+\! \frac{A_{2,i,k}^{(r)}E_{2,i,k}^{(r)}}{(d_e^{(r)}({\bf{q}}_{S,i}))^2}\Big)+\! B_{1,i,k}^{(r)}J_{1,i,k}^{(r)}\\
     &\triangleq -h^{(r)}_{i,k}({\bf{q}}_{u,i})J_{1,i,k}^{(r)}+ B_{1,i,k}^{(r)}J_{1,i,k}^{(r)},
\end{aligned}
\end{equation}
where, for notational convenience, $h^{(r)}_{i,k}({\bf q}_{u,i})$ aggregates all the terms that depend on the UAV trajectory variables ${\bf q}_{S,i}$ and ${\bf q}_{J,i}$, which are considered mutually independent. Subsequently, to further decouple the product of variables $a_{i,k}$, $t_i$, and $h^{(r)}_{i,k}({\bf q}_{u,i})$, we apply the arithmetic mean inequality and obtain the concave function as 
\begin{equation}
\begin{aligned}
    & {a_{i,k}}{t_i}R_k({\bf{Q}}_u,{\bf t},{\bf a}) \\
    &\ge \frac{-J_{1,i,k}^{(r)}}{27C_{1,i,k}^{(r)}C_{2,i,k}^{(r)}}\big({a_{i,k}}+C_{1,i,k}^{(r)}{t_i}+C_{2,i,k}^{(r)}h^{(r)}_{i,k}({\bf{q}}_{u,i})\big)^3\\
    &\quad-\frac{B_{1,i,k}^{(r)}J_{1,i,k}^{(r)}}{2F_{3,i,k}^{(r)}}(1-{a_{i,k}})^2-\frac{B_{1,i,k}^{(r)}J_{1,i,k}^{(r)}F_{3,i,k}^{(r)}}{2}{t^2_i}\\
    &\quad+ B_{1,i,k}^{(r)}J_{1,i,k}^{(r)}{t_i}\\
    &\triangleq s^{(r)}_{i,k}({\bf{Q}}_{u},{\bf t},{\bf a}), \label{app_hov}
\end{aligned}
\end{equation}
where $C_{1,i,k}^{(r)}=\frac{a_{i,k}}{t_i}$, $C_{2,i,k}^{(r)}=\frac{a_{i,k}}{h^{(r)}_{i,k}({\bf{q}}_{u,i})}$ and $F_{3,i,k}^{(r)}=\frac{1-{a_{i,k}}}{t_i}$ are positive constants. The equality holds when $({\bf{Q}}_{u},{\bf t},{\bf a})=({\bf{Q}}^{(r)}_{u},{\bf t}^{(r)},{\bf a}^{(r)})$.

\textit{2) Flying Period:} The upper limit of the integration over the flight segment, i.e., the flight time, involves the term $\Delta t_{i,j}=\frac{\max \{d_{i,j}({\bf Q}_S),d_{i,j}({\bf Q}_J)\}}{V}$. To avoid the non-smoothness caused by $\max\{\cdot\}$ operator in the flight time expression, we adopt an equivalent linear reformulation based on a binary indicator $\lambda \in \{0,1\}$, and define $d_{\text{max}}({\bf Q}_{u,i,j}) = (1 - \lambda)\, d_{i,j}({\bf Q}_S) + \lambda\, d_{i,j}({\bf Q}_J)$, where $\lambda$ is determined as
\begin{equation}
\lambda =
\begin{cases}
0, & \text{if } d_{i,j}({\bf Q}_S) \ge d_{i,j}({\bf Q}_J), \\
1, & \text{otherwise}.
\end{cases}
\end{equation}
As a result, the flight time can be expressed as $\Delta t_{i,j} = \frac{d_{\text{max}}({\bf Q}_{u,i,j})}{V}$.

Similar to the hovering period, we introduce another judgment matrix $J_{1,i,j,k}^{(r)}$ to address non-smooth operator ${\left[\cdot \right]^ + }$, which is expressed as 
\begin{equation}
    J_{1,i,j,k}^{(r)} = \begin{cases}
0, & \text{when  } R^{(r)}_{S,k}({{\bf{q}}^{(r)}_{u,i,j}}(z))\!-\!R^{(r)}_{S,e}({{\bf{q}}^{(r)}_{u,i,j}}(z)) \!< \!0 ,\\
1, & \text{when  } R^{(r)}_{S,k}({{\bf{q}}^{(r)}_{u,i,j}}(z))\!-\!R^{(r)}_{S,e}({{\bf{q}}^{(r)}_{u,i,j}}(z))\! \ge \!0.
\end{cases}
\end{equation}

The received SNRs at ground node $m\in\{k,e\}$ during the flying period can be expressed as
\begin{equation}
    {\gamma_{i,j,m}}(z)\!=\!\frac{\beta_0P_S d^2_m({{\bf{q}}_{J,i,j}}(z))}{\beta_0P_J d^2_m({{\bf{q}}_{S,i,j}}(z))\!+\!\sigma_k^2 d^2_m({{\bf{q}}_{S,i,j}}(z)) d^2_m({{\bf{q}}_{J,i,j}}(z))}.
\end{equation}

Then, noting that the function $f(u)=-\log_2(1+u)-\sqrt{V(u)}$ is convex to $u$, and $g(v)=\log_2(1+v)$ is convex to $\frac{1}{v}$, we can derive 
\begin{equation}
    \begin{aligned}
        &\int_0^1 {R_k({{\bf{q}}_{u,i,j}}(z))d{z}} \\
        &\ge \int_0^1
        \Big(-\frac{A_{1,i,j,k}^{(r)}}{\gamma_{i,j,k}(z)}-A_{2,i,j,k}^{(r)}{\gamma_{i,j,e}(z)} + B_{1,i,j,k}^{(r)}\Big)J_{1,i,j,k}^{(r)}dz, \label{approxi_fly}
    \end{aligned}
\end{equation}
where $A_{1,i,j,k}^{(r)}=\frac{({\gamma^{(r)}_{i,j,k}}(z))^2}{\ln2({\gamma^{(r)}_{i,j,k}(z)}+1)}$, $A_{2,i,j,k}^{(r)}=\frac{1}{\ln2({\gamma^{(r)}_{i,j,e}(z)}+1)}$, and $B_{1,i,j,k}^{(r)}=\frac{A_{1,i,j,k}^{(r)}}{\gamma^{(r)}_{i,j,k}(z)} + A_{2,i,j,k}^{(r)}{\gamma^{(r)}_{i,j,e}(z)}+R^{(r)}_k({{\bf{q}}_{u,i,j}}(z))$ are positive constants. 

Similar to the approximation at hovering points, we apply the arithmetic mean–geometric mean inequality to decouple the multiplicative terms in the SNR expressions as follows
\begin{subequations}
\begin{align}
&\frac{1}{{\gamma_{i,j,k}(z)}}\le \frac{P_J}{P_S}(\frac{F_{1,i,k}^{(r)}}{2 d^4_k({{\bf{q}}_{J,i,j}}(z))}+\frac{d^4_k({{\bf{q}}_{S,i,j}}(z))}{2F_{1,i,k}^{(r)}}),\\
&{\gamma_{i,j,e}(z)} \le E_{1,i,k}^{(r)}(\!\frac{F_{2,i,k}^{(r)}}{2 d^4_e({{\bf{q}}_{S,i,j}}(\!{z}\!)\!)}\!\!+\!\!\frac{d^4_e({{\bf{q}}_{J,i,j}}(\!{z}\!)\!)}{2F_{2,i,k}^{(r)}}\!)\!+\!\!\frac{E_{2,i,k}^{(r)}}{d^2_e({{\bf{q}}_{S,i,j}}(\!{z}\!)\!)},
\end{align}
\end{subequations}
where $F_{1,i,j,k}^{(r)}\!\!=\!d_k^2({{\bf{q}}^{(r)}_{S,i,j}}(z)) d_k^2({{\bf{q}}^{(r)}_{J,i,j}}(z))$, $F_{2,i,j,k}^{(r)}\!\! =\! d_e^2({{\bf{q}}^{(r)}_{S,i,j}}(z)) d_e^2({{\bf{q}}^{(r)}_{J,i,j}}(z))$, $E_{1,i,j,k}^{(r)} \!\! =\! \frac{\beta_0^2 P_S P_J}{\left( \beta_0 P_J + \sigma_e^2 d_e^2({\bf{q}}_{J,i,j}^{(r)}(z)) \right)^2}$, and $E_{2,i,j,k}^{(r)} \!\! =\!\! -\!E_{1,i,j,k}^{(r)} d_e^2({\bf{q}}_{J,i,j}^{(r)}(z)) \!+\! \frac{\beta_0 P_S d_e^2({\bf{q}}_{J,i,j}^{(r)}(z))}{\beta_0 P_J + \sigma_e^2 d_e^2({\bf{q}}_{J,i,j}^{(r)}(z))}$ are positive constants. 

By applying the arithmetic mean inequality, the function $\int_0^1 {R_k({{\bf{q}}_{u,i,j}}(z))d{z}}$ can be further approximated in \eqref{app-fly}, where to ensure convexity, we approximate the non-convex terms $\frac{1}{d_k^2({\bf q}_{S,i,j}(z))}$ and $\frac{1}{d_e^2({\bf q}_{S,i,j}(z))}$ as follows
\begin{align}
    &{d^2_m({\bf{q}}_{J,i,j}(z))}  \nonumber\\
    &\ge  2({\bf{q}}^{(r)}_{J,i,j}(z)\!-\!{{\bf{w}}_m})^T({\bf{q}}_{J,i,j}(z)\!-\!{{\bf{w}}_m})\!-\!{\left\| {{\bf{q}}^{(r)}_{J,i,j}(z)-\! {{\bf{w}}_m}} \right\|}^2 \nonumber\\
    &\triangleq (d_m^{(r)}({\bf{q}}^{(r)}_{J,i,j}(z)))^2, m\in\{k,e\},
\end{align}
where the corresponding linear constraints are also introduced as follows
\begin{equation}
    2({\bf{q}}^{(r)}_{J,i,j}(z)\!-\!{{\bf{w}}_m})^T({\bf{q}}_{J,i,j}(z)\!-\!{{\bf{w}}_m})\!-\!\!{\left\| {{\bf{q}}^{(r)}_{J,i,j}(z)\!- \!\!{{\bf{w}}_m}} \right\|}^2
     \!\!\!\geq\!\! 0.
\end{equation}

Since the constraints must hold for all $z \in [0,1]$, we adopt the same method used in the anti-collision constraint in \eqref{collision-min} by introducing $\min_{z \in [0,1]}(d_k^{(r)}({\bf{q}}^{(r)}_{J,i,j}(z)))^2 \triangleq  \kappa^{(r)}_{J,k,i,j}\ge0$ and $\min_{z \in [0,1]}(d_e^{(r)}({\bf{q}}^{(r)}_{S,i,j}(z)))^2\triangleq  \kappa^{(r)}_{S,e,i,j}\ge0$.

\begin{figure*}[t!]
\begin{align}
&\int_0^1 {R_k({{\bf{q}}_{u,i,j}}(z))d{z}} \nonumber\\
     &\ge \int_0^1 J_{1,i,j,k}^{(r)}\Big(-A_{1,i,j,k}^{(r)} \frac{\sigma_k^2 d_k^2({{\bf{q}}_{S,i,j}}(z))}{\beta_0 P_S}- \frac{P_J A_{1,i,j,k}^{(r)}}{P_S}\big(\frac{F_{1,i,j,k}^{(r)}}{2(d^{(r)}_k({{\bf{q}}_{S,i,j}}(z)))^4}+\frac{d^4_k({{\bf{q}}_{S,i,j}}(z))}{2F_{1,i,j,k}^{(r)}}\big) \nonumber\\
     &\quad -A_{2,i,j,k}^{(r)}E_{1,i,j,k}^{(r)}\big(\frac{F_{2,i,j,k}^{(r)}}{2(d^{(r)}_e({{\bf{q}}_{S,i,j}}(z)))^4}+\frac{d^4_e({{\bf{q}}_{S,i,j}}(z))}{2F_{2,i,j,k}^{(r)}}\big) -\frac{A_{2,i,j,k}^{(r)}E_{2,i,j,k}^{(r)}}{(d^{(r)}_e({{\bf{q}}_{S,i,j}}(z)))^2} + B_{1,i,j,k}^{(r)} \Big) dz \nonumber\\
     &\triangleq \int_0^1 -h^{(r)}_{i,j,k}({\bf{q}}_{u,i,j}(z))+ B_{1,i,j,k}^{(r)}J_{1,i,j,k}^{(r)} dz,\label{app-fly}
\end{align}
\hrulefill
\end{figure*}

Next, by applying the arithmetic mean inequality, we can further decompose the terms and derive the concave function as follows
\begin{equation}
\begin{aligned}
    &{a_{i,j,k}\frac{d_{\text{max}}({\bf Q}_{u,i,j})}{V}}\int_0^1 R_k({{\bf{q}}_{u,i,j}}(z))d{z}\\
    &\ge \frac{-1}{27VC_{1,i,j,k}^{(r)}C_{2,i,j,k}^{(r)}}\Big({a_{i,j,k}}+C_{1,i,j,k}^{(r)}d_{\text{max}}({\bf Q}_{u,i,j})\\
    &\quad +C_{2,i,j,k}^{(r)}\!\int_0^1 \!\!\! h^{(r)}_{i,j,k}({\bf{q}}_{u,i,j}(z))dz\Big)^3\!\!-\!\frac{B_{2,i,j,k}^{(r)}}{2F_{3,i,j,k}^{(r)}}(1\!-{a_{i,j,k}})^2\\
    &\quad-\frac{B_{2,i,j,k}^{(r)}F_{3,i,j,k}^{(r)}}{2}d^2_{i,j,max}+ B_{2,i,j,k}^{(r)}d^{(r)}_{\text{max}}({\bf Q}_{u,i,j})\\
    &\triangleq v^{(r)}_{i,j,k}({\bf{Q}}_{u},{\bf a}), \label{app_fly}
\end{aligned}
\end{equation}
where $B_{2,i,j,k}^{(r)}=\frac{1}{V}\int^1_0 B_{1,i,j,k}^{(r)}J_{1,i,j,k}^{(r)}dz$, $C_{1,i,j,k}^{(r)}=\frac{a_{i,j,k}}{\big(d_{\text{max}}({\bf Q}_{u,i,j})\big)^{(r)}}$, $C_{2,i,j,k}^{(r)}=\frac{a_{i,j,k}}{\int_0^1 h^{(r)}_{i,j,k}({\bf{q}}_{u,i,j}(z))dz}$ and $F_{3,i,j,k}^{(r)}=\frac{1-{a_{i,j,k}}}{\big(d_{\text{max}}({\bf Q}_{u,i,j})\big)^{(r)}}$ are positive constants. The equality holds when $({\bf{Q}}_{u},{\bf a})=({\bf{Q}}^{(r)}_{u},{\bf a}^{(r)})$. 

Additionally, as $d_{\text{max}}({\bf Q}_{u,i,j}) = (1 - \lambda)\, d_{i,j}({\bf Q}_S) + \lambda\, d_{i,j}({\bf Q}_J)$ and ${d_{i,j}}({\bf{Q}}_u)$ is convex in ${\bf{q}}_{u,i,j}$ and ${\bf{q}}_{u,i,j+1}$, it can be approximated as
\begin{equation}
    \begin{aligned}
        {d_{i,j}}({\bf{Q}}_u) 
        \ge & \frac{1}{{{d_{i,j}}({{\bf{Q}}_u^{(r)}})}}(x_{u,i,j + 1}^{(r)}\! -\! x_{u,i,j}^{(r)})({x_{u,i,j + 1}} \!-\! {x_{u,i,j}}) \\
        & + \!\frac{1}{{{d_{i,j}}({{\bf{Q}}_u^{(r)}})}}(y_{u,i,j + 1}^{(r)} \!-\! y_{u,i,j}^{(r)})({y_{u,i,j + 1}} \!-\! {y_{u,i,j}})\\
 \triangleq& d_{i,j}^{(r)}({\bf{Q}}_u),
    \end{aligned}
\end{equation}
where the equality holds when ${\bf{Q}}_u={\bf{Q}}^{(r)}_u$. Thus, we can obtain $d^{(r)}_{\text{max}}({\bf Q}_{u,i,j})\triangleq(1-\lambda) d_{i,j}^{(r)}({\bf Q}_S)+\lambda d_{i,j}^{(r)}({\bf Q}_J)$.

Combining the approximation of hovering period and flying period, i.e., (\ref{app_hov}) and (\ref{app_fly}), the concave approximation of the entire period is expressed as
\begin{equation}
    \begin{aligned}
        &{U_k}({\bf{Q}}_{u},{\bf t},{\bf a}) \\
        \ge& \sum\limits_{i = 1}^K  {s^{(r)}_{i,k}({\bf{Q}}_{u},{\bf t},{\bf a})}  + \sum\limits_{i = 0}^K \sum\limits_{j = 0}^N  {v^{(r)}_{i,j,k}({\bf{Q}}_{u},{\bf a})}\\
        \triangleq & U_k^{(r)}({\bf{Q}}_{u},{\bf t},{\bf a}).
    \end{aligned}
\end{equation}

Substituting $U_k({\bf{Q}}_{u},{\bf t},{\bf a})$ by $U_k^{(r)}({\bf{Q}}_{u},{\bf t},{\bf a})$, we can transform Problem (P3) into 
\begin{subequations}\label{P4}
    \begin{alignat}{2}
        ({\rm P3}):&\max \limits_{{\bf{Q}}_{u},{\bf t},{\bf a}} && U   \\
       &\quad\!\!\!\mathrm{s.t.:} && U_k^{(r)}({\bf{Q}}_{u},{\bf t},{\bf a}) \ge U, \\
     &&&\kappa^{(r)}_{J,k,i,j}\ge0,\kappa^{(r)}_{S,e,i,j}\ge0,\\
       &&&\eqref{con:initial final point2},\eqref{con:task time},\eqref{con:scheduling2_2},\eqref{con:duration},\eqref{con:scheduling1_3}, \eqref{con:collision6}.\nonumber
    \end{alignat}
\end{subequations}
Up to this point, we have constructed the convex problem (P3) that can be efficiently addressed utilizing convex optimization methods. Afterwards, we introduce an iterative algorithm to tackle this issue effectively.

\subsection{Iterative Algorithm}
As illustrated in Algorithm 1, we first initialize feasible values for the variables as $({\bf{Q}}^{(0)}_{u},{\bf t}^{(0)},{\bf a}^{(0)})$ and set the iteration index to $r=0$. For the initialization of the UAV-S trajectory, we solve the Traveling Salesman Problem (TSP) to determine the shortest path passing through all GUs. The hovering points are positioned directly above the GU locations, with $u$ turning points evenly distributed between adjacent hovering points, resulting in the initial trajectory ${\bf{Q}}^{(0)}_S$. Similarly, the UAV-J trajectory, ${\bf{Q}}^{(0)}_J$, is initialized as a straight line connecting its starting and ending points, where the trajectory points are evenly distributed along the line segment. Subsequently, the total flight time is calculated by summing the maximum travel times between UAV-S and UAV-J for each segment. The remaining task time is then evenly distributed among all GUs, leading to the initialization of the hovering durations ${\bf{t}}^{(0)}$. Additionally, the communication scheduling is initialized as ${\bf{a}}^{(0)} = \frac{1}{K}$, where each user's scheduling coefficient is equally distributed. Thus, a feasible initial point $({\bf{Q}}^{(0)}_{u},{\bf t}^{(0)},{\bf a}^{(0)})$ is obtained.

In the $r$-th iteration, the convex approximation $R_k^{(r)}({\bf{Q}}_{u},{\bf t},{\bf a})$ is constructed based on the current feasible point $({\bf{Q}}^{(r)}_{u},{\bf t}^{(r)},{\bf a}^{(r)})$, and the corresponding convex problem (P3) is formulated. Using the ellipsoid method, the optimal solution of (P3), denoted as $({\bf{Q}}^{(r+1)}_{u},{\bf t}^{(r+1)},{\bf a}^{(r+1)})$, is obtained. This solution is then used as the feasible point for the next iteration. The above process is repeated until the objective function converges. It is worth noting that each iteration involves $M_{var}=2NK^2+3K^2+6NK+6K+2N$ optimization variables and $M_{con}=6NK^2+8K^2+12NK+13K+6N+3$ constraints. Therefore, as stated in the analysis, the algorithm complexity is $\mathcal{O}(\varphi M_{var}^2(M_{var}^2+M_{var}M_{con}))$, where $\varphi$ denotes the total number of iterations, and $\varepsilon$ represents the convergence threshold.

\begin{algorithm}[!t]
\small
\caption{\bf{: Joint optimization of UAV trajectories and communication scheduling.}}
\begin{algorithmic}
\STATE \noindent{\bf{$\!\!\!\!\!\!$Initialization}} \\
\STATE   Set iteration index $r=0$ and initialize $({\bf{Q}}^{(0)}_{u},{\bf t}^{(0)},{\bf a}^{(0)})$. 

 \STATE \noindent{ \bf{$\!\!\!\!\!\!\!$Iteration}} \\

 \STATE \noindent{\bf{a)}} Establish convex approximation $U_k^{(r)}({\bf{Q}}_{u},{\bf t},{\bf a})$ according to $({\bf{Q}}^{(r)}_{u},{\bf t}^{(r)},{\bf a}^{(r)})$;\\

 \STATE \noindent{\bf{b)}} Address the convex problem (P3);\\

 \STATE \noindent{\bf{c)}} \noindent{\bf{If}} ~ the iteration difference of the objective is less than $\varepsilon$\\
 
 \STATE ~~~~~~~~ Acquire the optimal solution as $({\bf{Q}}^*_{u},{\bf t}^*,{\bf a}^*)$;\\
 
 \STATE ~~~~~~~~ stop iteration;\\

 \STATE ~~~ \noindent{\bf{Else}} \\

 \STATE ~~~~~~~~ $r=r+1$;\\
 
 \STATE ~~~~~~~~ $({\bf{Q}}^{(r+1)}_{u},{\bf t}^{(r+1)},{\bf a}^{(r+1)})=({\bf{Q}}^*_{u},{\bf t}^*,{\bf a}^*)$;\\
 
 \STATE ~~~~~~~~ return to \bf{a)}.

\end{algorithmic}
\label{algorithm1}
\end{algorithm}

\section{Numerical Results}
\label{sec:numerical result}
In this section, we present numerical results to evaluate the proposed design in terms of the minimum secrecy throughput among GUs. We consider a square operational area of 500 meters in width, where GUs and Eve are randomly distributed. The initial and final locations of UAVs are set to be ${\bf q}_{S,I} = [450,450]$, ${\bf q}_{S,F} = [450,50]$, ${\bf q}_{J,I} = [50,450]$, and ${\bf q}_{J,F} = [50,50]$. Unless otherwise stated, the simulation parameters are set as follows: $K=4$, $N=1$, $T = 150$s, $V = 10$m/s, $d_{min} = 3$m, $P_S = 10$mW, $P_J = 1$mW, $\beta_0 = -30$dB, $\sigma_k^2 = -80$dBm, and $\sigma_e^2 = -80$dBm.

To demonstrate the effectiveness of the proposed design, we compare it against the following benchmarks:
\begin{itemize} 
\item \textbf{Single-UAV SHF}: In the absence of UAV-J, the trajectory and communication scheduling of UAV-S are jointly optimized based on the SHF structure proposed in \cite{Yuan2021}.

\item \textbf{TD-SCP design}: The task time is divided into $N_0$ sufficiently small and equal-length time slots $\Delta=\frac{T}{N_0}$, ensuring that the locations of UAVs remain approximately unchanged within each time slot. 
\end{itemize}

\begin{figure}[h]
	\centering
    \vspace{-.2cm}
\includegraphics[width=7.5 cm, trim = 5 10 5 10]{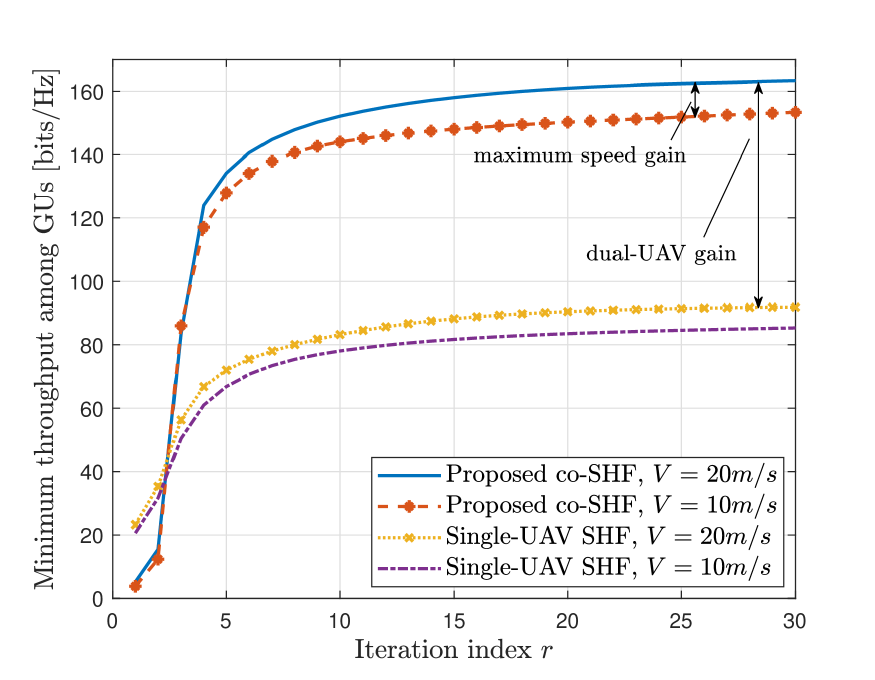}
\caption{Convergence of the proposed and Single-UAV designs under different UAV speed constraints.\label{convergence}}
\vspace{-.2cm}
\end{figure}
We first evaluate the convergence behavior of the proposed co-SHF design and the Single-UAV SHF design under different maximum speed constraints, as shown in Fig.~\ref{convergence}. All schemes exhibit rapid and stable convergence, with the minimum secrecy throughput increasing monotonically and stabilizing within around 15 iterations. For a fixed UAV speed, the proposed co-SHF solution achieves significantly higher minimum secrecy throughput than the Single-UAV SHF design, demonstrating the benefit of deploying a cooperative jammer UAV. Moreover, increasing the maximum UAV speed from 10 m/s to 20 m/s yields noticeable performance gains for both designs, as higher mobility allows more mission time to be allocated to hovering rather than flying, thereby improving both transmission and jamming effectiveness.

\begin{figure}[!t]
	\centering
    \includegraphics[width=7.5 cm, trim = 5 10 5 10]{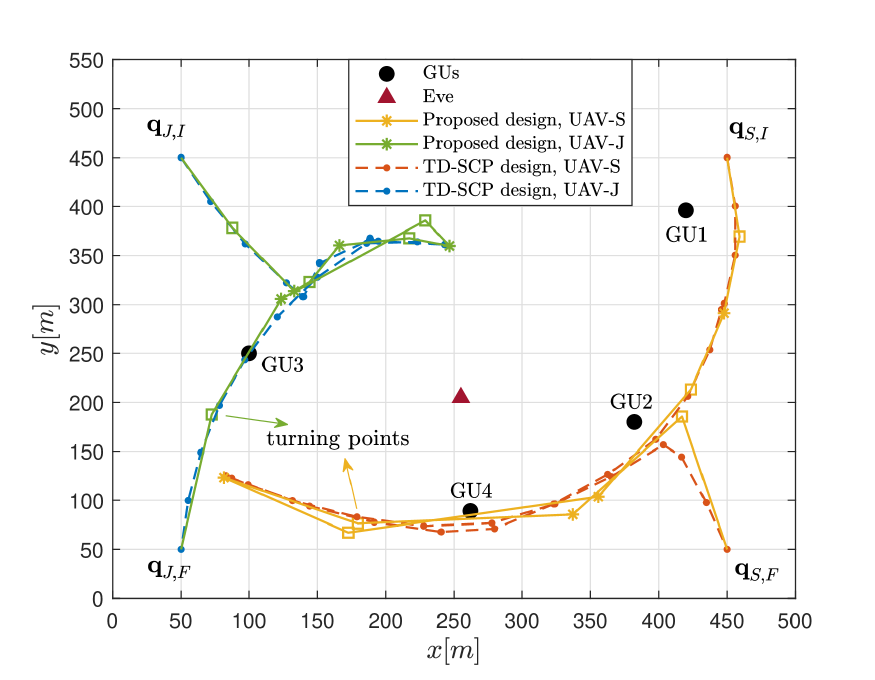}
\caption{UAV trajectories under the proposed and TD-SCP schemes.\label{traj}}
\end{figure}

We analyze the optimized trajectories of the two UAVs under different design schemes, as shown in Fig.~\ref{traj}. In both schemes, UAV-S hovers near the GUs to support communication, while UAV-J maneuvers around Eve to generate jamming signals. Notably, UAV-J maintains a significantly larger distance from Eve in both schemes, as Eve is located relatively close to several GUs and excessive proximity may result in unintended interference on legitimate links. Therefore, the UAV-J trajectory reflects a critical tradeoff between strengthening secrecy and preserving communication reliability. Moreover, the UAV-S trajectory in the proposed design consists of nearly straight flight segments with minimal directional changes, indicating that more mission time is allocated to effective hovering rather than to frequent maneuvering, thereby improving communication and jamming efficiency. 

\begin{figure}[!t]
	\centering
\includegraphics[width=7.5 cm, trim = 5 10 5 10]{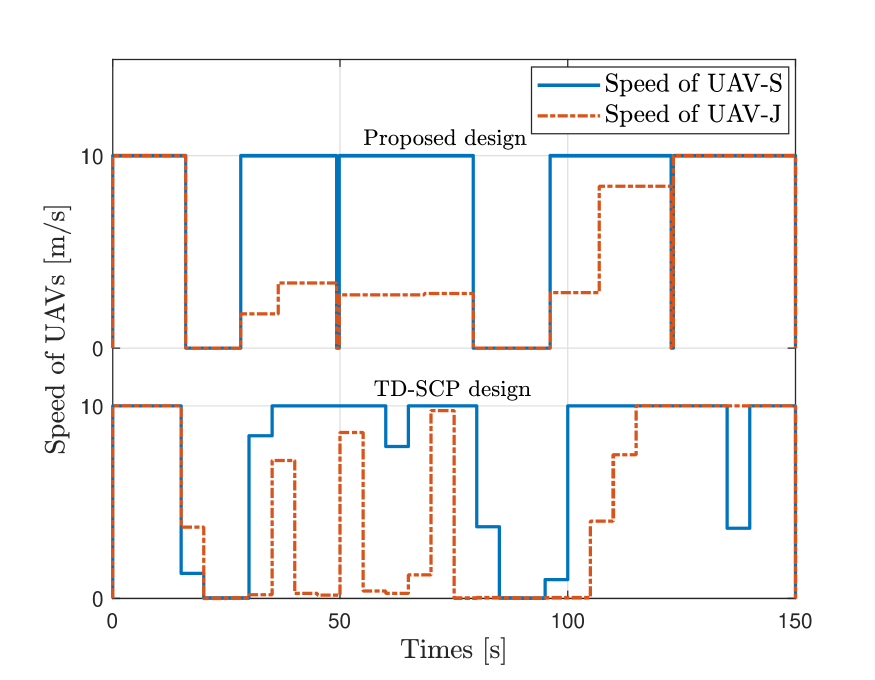}
\caption{Speed shifts of UAVs during the mission under the proposed and TD-SCP designs.\label{speed}}
\end{figure}

To gain further insight into the speed variation of the UAVs, Fig.~\ref{speed} presents their speed profiles during the mission under the proposed and TD-SCP designs. Both schemes exhibit alternating flight and hovering phases, with multiple time intervals showing simultaneous flight or simultaneous hovering. Notably, this characterization emerges in the TD-SCP scheme despite the absence of any explicit structural constraint, which further supports the validity of the proposed co-SHF structure. In the proposed design, at least one UAV consistently travels at maximum speed during flight phases, and in certain intervals, both UAVs fly at full speed simultaneously. This indicates a time-efficient movement strategy that reserves more mission time for critical hovering tasks such as data transmission and jamming. In contrast, the TD-SCP scheme includes multiple non-hovering intervals where neither UAV flies at maximum speed, suggesting inefficient time utilization and contributing to the performance gap. Moreover, the hovering durations at the second and fourth co-hovering point pairs in the proposed design are nearly zero, which is consistent with Proposition~\ref{prop2} that no more than K co-hovering pairs are required. 

\begin{figure}[!t]
	\centering
\includegraphics[width=7.5 cm, trim = 5 10 5 10]{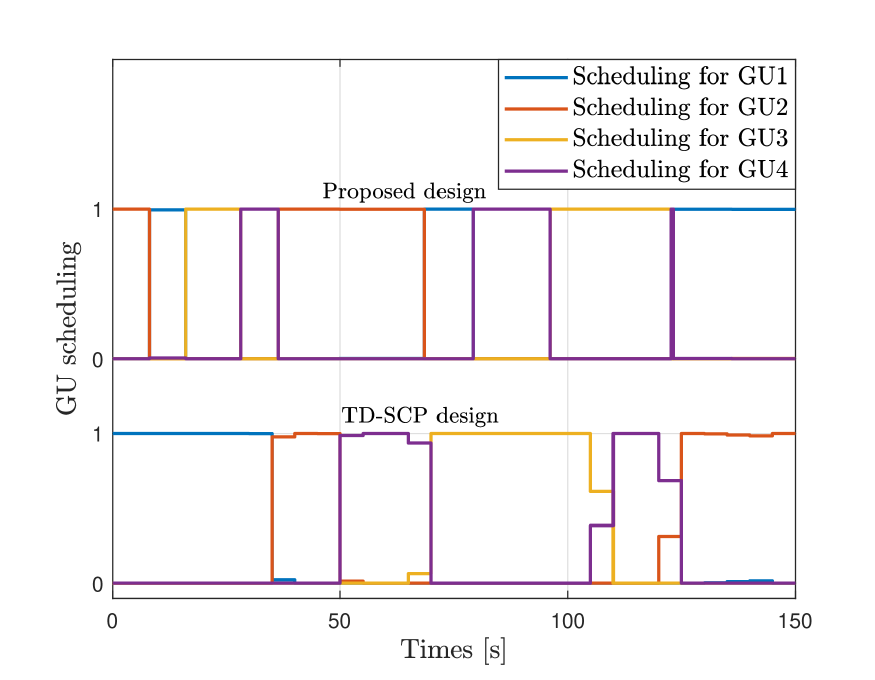}
\caption{Communication scheduling during the mission under the proposed and TD-SCP designs.\label{scheduling}}
\end{figure}
Beyond movement coordination, it is also essential to investigate how these schemes schedule users over time. Fig.~\ref{scheduling} presents the scheduling behavior of the proposed and TD-SCP designs during the mission. In the proposed design, the scheduling profile exhibits a binary nature, where exactly one GU is scheduled at each time instant. This strictly satisfies the binary scheduling constraint, reflecting the effectiveness of the proposed scheduling strategy. In contrast, the TD-SCP design displays several intervals where scheduling values deviate from 0 and 1, indicating partial or ambiguous user assignments. Such fractional scheduling does not strictly satisfy the binary scheduling constraint and may result in implementation discrepancies, as real systems require integer-valued assignments. Moreover, taken together with the trajectories shown in Fig.~\ref{traj}, it can be seen that both schemes tend to schedule GUs when UAV-S is geographically closer, which aligns with expectations. However, the proposed design adapts the scheduling durations for each GU more flexibly, aiming to maximize the minimum throughput. This leads to longer service durations for users with higher interference exposure, such as GU 2 and GU 4, and reflects a fairness-aware optimization. Conversely, the TD-SCP design allocates similar scheduling times to all GUs, which compromises fairness and can result in degraded system performance, especially for vulnerable users.

\begin{figure}[!t]
	\centering
\includegraphics[width=7.5 cm, trim = 5 10 5 10]{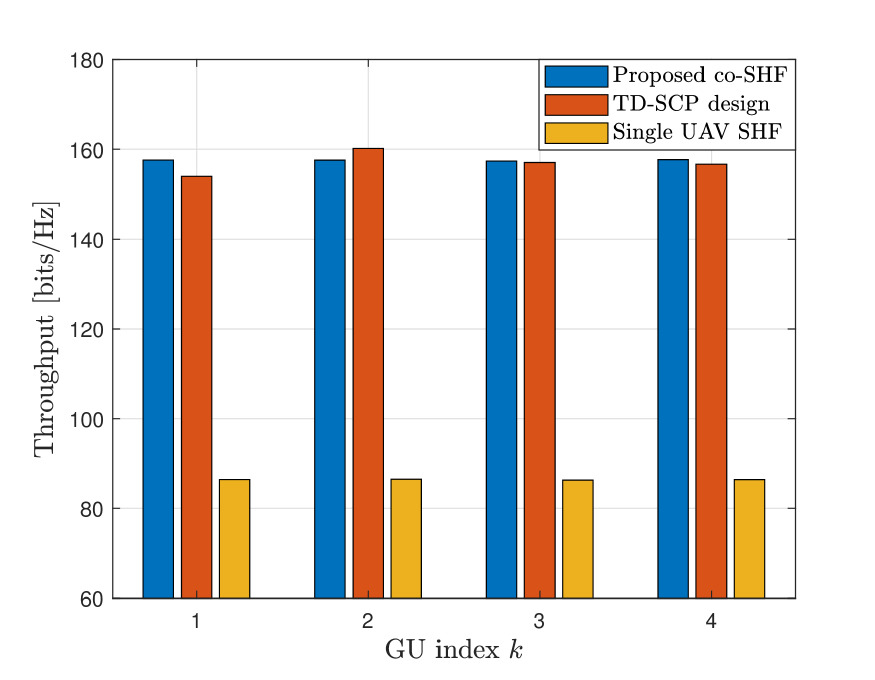}
\caption{Secrecy throughput of each GU under different schemes.\label{fair}}
\end{figure}
The scheduling behavior observed in Fig.~\ref{scheduling} has a direct impact on per-user secrecy throughput, as shown in Fig.~\ref{fair}. The proposed scheme ensures relatively uniform throughput among all GUs, indicating that no user is overly favored or neglected. This aligns with the optimization objective of maximizing the minimum throughput across users. In contrast, the TD-SCP scheme exhibits significant disparity in user throughput. Notably, GU~1 achieves markedly higher throughput than the others due to its favorable position farther from Eve, resulting in less interference and stronger channel quality. However, the TD-SCP scheme fails to compensate for such interference-related differences, as it allocates scheduling time uniformly among users regardless of their relative vulnerability. This behavior, as observed in Fig.~\ref{scheduling}, leads to a mismatch between user needs and scheduling resources, ultimately resulting in unfairness. Additionally, the Single-UAV design yields the lowest overall throughput, further highlighting the benefits of using a dual-UAV cooperative structure to enhance both communication performance and secrecy protection.

\begin{figure}[!t]
	\centering
    \includegraphics[width=7.5 cm, trim = 5 10 5 10]{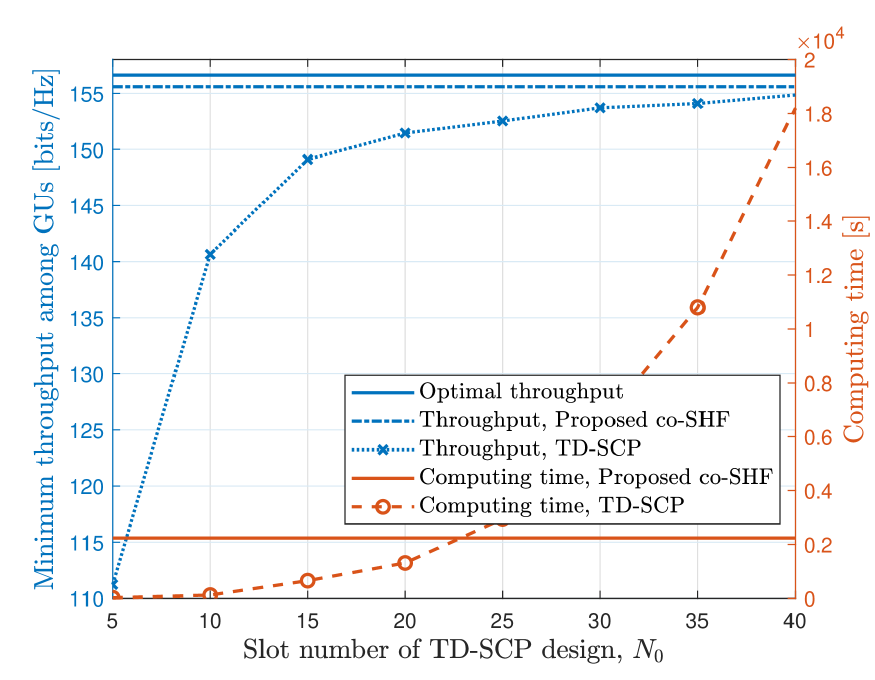}
\caption{Secrecy throughput and computing time versus slot number $N_0$ of the TD-SCP design, compared with the proposed design.\label{quantization}}
\end{figure}
Fig.~\ref{quantization} illustrates the impact of the slot number $N_0$ on secrecy performance and computing time of the TD-SCP design in comparison with the proposed method. As $N_0$ increases, TD-SCP gradually improves and approaches the curve of the proposed design.
However, our method remains ahead across a broad range of $N_0$. Around $N_0\approx 22$, the two methods exhibit comparable wall-clock time, yet TD-SCP still shows a clear gap in the minimum secrecy throughput. TD-SCP reaches a similar throughput only at about $N_0\approx 40$, where its wall-clock time is about \emph{nine times} that of our method. In terms of computational efficiency, the runtime of TD-SCP grows rapidly with $N_0$, whereas the proposed design with a single turning point $N=1$ maintains a much lower computing time since its optimization variables are tied to the number of co-hovering pairs, i.e., the number of GUs, rather than to slot density. Overall, the proposed design achieves comparable or better secrecy performance than high-resolution TD-SCP while requiring only a fraction of its runtime. It is worth mentioning that this runtime ratio reflects empirical wall-clock measurements in our simulations, rather than a theoretical complexity bound.

\begin{figure}[!t]
	\centering
\includegraphics[width=7.5 cm, trim = 5 10 5 10]{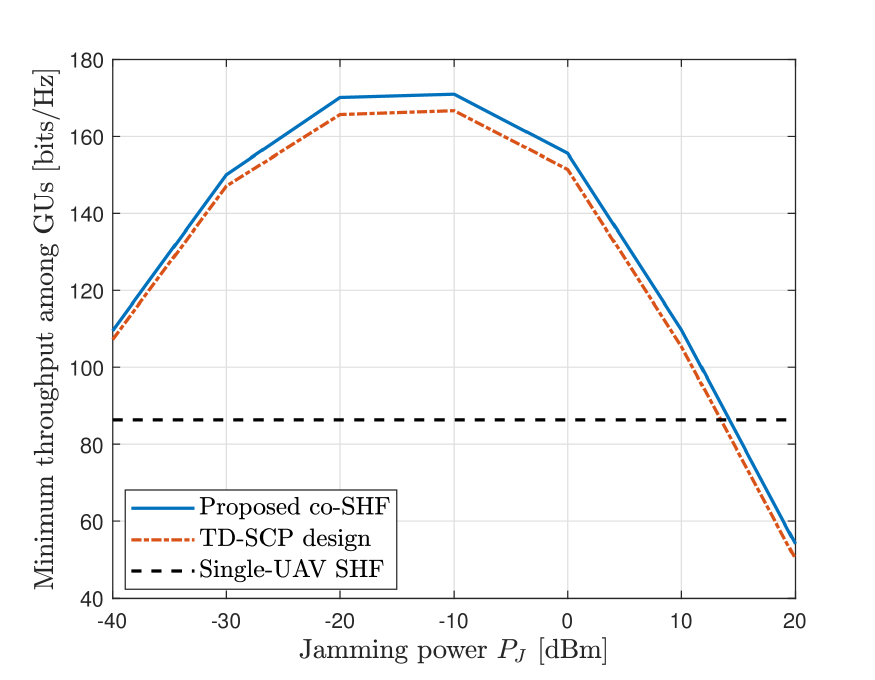}
\caption{Minimum secrecy throughput among GUs versus jamming power $P_J$ under different schemes.\label{power}}
\end{figure}


Fig.~\ref{power} shows the effect of jamming power $P_J$ on the minimum secrecy throughput under different schemes. The proposed design consistently outperforms both the TD-SCP and Single-UAV schemes across the entire range of $P_J$, demonstrating its ability to adapt to varying interference power levels. For all schemes, the throughput initially increases with $P_J$, as stronger jamming more effectively degrades the eavesdropper’s channel. However, beyond a certain threshold, further increases in jamming power begin to impair the legitimate communication link between UAV-S and the GUs, leading to a decline in the minimum secrecy throughput. These results underscore the importance of optimizing $P_J$ to balance secrecy enhancement and communication performance in cooperative UAV systems. 

\section{Conclusion}
\label{sec:conclusion}
This paper studied a dual-UAV secure communication system in which a data UAV serves multiple users while a cooperative jammer protects against a ground eavesdropper. We maximized the minimum secrecy throughput by jointly designing trajectories and communication scheduling under UAV mobility, anti-collision, and communication scheduling constraints. A central contribution of this work lies in proposing, for the first time, a structural characterization of the optimal dual-UAV trajectory design for the considered scenario. We rigorously proved that the UAVs must follow a collaborative SHF structure in which the two UAVs visit a finite set of synchronized co-hovering point pairs and, during each flight segment, at least one UAV moves at the maximum speed.  The structure enables an equivalent finite-dimensional reformulation with a limited number of variables. Building on this representation, we developed a successive convex approximation method that converges and runs with low computational burden. Simulation results verify that  the proposed design preserves or improves the minimum secrecy while achieving markedly shorter wall-clock time than time discretization. The effectiveness of coordinated trajectory design and the benefit of jammer UAV introduction are also demonstrated. 

It is worth noting that the proposed design exhibits strong flexibility and applicability to a wider range of multi-UAV communication systems with diverse mission profiles.  First, scaling from the dual-UAV case to systems with multiple cooperating agents, each with distinct yet interdependent roles, induces an exponential growth in trajectory and scheduling variables, under which the co-SHF structure demonstrates even greater advantages in preserving optimization efficiency. Second, the proposed trajectory structure is not confined to jamming-communication coordination. It is broadly applicable to a wider class of \textit{heterogeneous coordinated UAV} networks, where multiple UAVs perform distinct but temporally and spatially coupled tasks toward a common objective. Further directions include more complex air-to-ground propagation models, imperfect channel information, mobile users, and multiple eavesdroppers, all of which can be layered on the co-SHF foundation without changing the core methodology.

\appendices

\bibliography{mybibfile}






\end{document}